\documentclass[reqno]{amsart}

\usepackage{fullpage,amsrefs,amssymb,color}%,showlabels}
\usepackage[colorlinks=false]{hyperref}

\usepackage{mathrsfs} % to write scri with mathscr

\title{On the asymptotic behavior of static perfect fluids} 

\keywords{Einstein--Euler equations, spherical symmetry, static, asymptotics, dynamical system, quasi-asymptotically flat, asymptotically conical, ADM mass\\
Appears in \emph{Annales Henri Poincar\'e} (2019), DOI \url{https://doi.org/10.1007/s00023-018-00758-z}}

\author[L.~Andersson]{Lars~Andersson}
\address{Albert Einstein Institute
(Max Planck Institute for Gravitational Physics),
Am M\"uhlenberg 1,
14476 Potsdam, Germany}
\email{\href{mailto:laan@aei.mpg.de}{laan@aei.mpg.de}}
\author[A.~Burtscher]{Annegret Y.\ Burtscher} 
\address{Department of Mathematics, Rutgers, The State University of New Jersey, 110 Frelinghuysen Rd., Piscataway, NJ 08854-8019, United States of America\\ \newline
Current address: Department of Mathematics, Radboud University, PO Box 9010, Postvak 59, 6500 GL Nijmegen, The Netherlands}
\email{\href{mailto:burtscher@math.ru.nl}{burtscher@math.ru.nl}}
\numberwithin{equation}{section}

\theoremstyle{plain}
\newtheorem{theorem}{Theorem}[section]
\newtheorem{corollary}[theorem]{Corollary}
\newtheorem{proposition}[theorem]{Proposition}
\newtheorem{lemma}[theorem]{Lemma}

\theoremstyle{definition}
\newtheorem{definition}[theorem]{Definition}

\theoremstyle{remark}
\newtheorem{remark}[theorem]{Remark}

\def\RR{\mathbb{R}}

\def\gab{g^{(\alpha,\beta)}}

\newcommand{\G}[2]{\ensuremath{\Gamma^{#1}_{#2}}} % Christoffel symbols
\begin{document}

%%%%% TITLE AND ABSTRACT %%%%%

\begin{abstract}
 Static spherically symmetric solutions to the Einstein--Euler equations with prescribed central densities are known to exist, be unique and smooth for reasonable equations of state. Some criteria are also available to decide whether solutions have finite extent (stars with a vacuum exterior) or infinite extent. In the latter case, the matter extends globally with the density approaching zero at infinity. The asymptotic behavior largely depends on the equation of state of the fluid and is still poorly understood. While a few such unbounded solutions are known to be asymptotically flat with finite ADM mass, the vast majority are not. We provide a full geometric description of the asymptotic behavior of static spherically symmetric perfect fluid solutions with linear equations of state and polytropic-type equations of state with index $n>5$. In order to capture the asymptotic behavior, we introduce a notion of scaled quasi-asymptotic flatness, which encompasses the notion of asymptotic conicality. In particular, these spacetimes are asymptotically simple.
\end{abstract}

\maketitle

%%%%%%%%%%%%%%%%%%%%%%%%%%%%%%%%%%%%%%%%%%%%%%%%%%%%%%%%%%%%%%%%%%%%%%%%%%%%%%%%%%%%%%%%
%%%%%% SECTION 1 %%%%%%%%%%%%%%%%%%%%%%%%%%%%%%%%%%%%%%%%%%%%%%%%%%%%%%%%%%%%%%%%%%%%%%%
%%%%%%%%%%%%%%%%%%%%%%%%%%%%%%%%%%%%%%%%%%%%%%%%%%%%%%%%%%%%%%%%%%%%%%%%%%%%%%%%%%%%%%%%

%%%%%% INTRODUCTION %%%%%%

\section{Introduction}
\label{intro}

 Perfect fluids in general relativity are described by the Einstein--Euler equations, i.e.,
\begin{align}\label{EE}
 G_{\alpha\beta} = 8 \pi T_{\alpha\beta}, \qquad \nabla_\alpha T^{\alpha\beta} = 0,
\end{align}
 where $G_{\alpha\beta} = R_{\alpha\beta} - \frac{1}{2} \, R \, g_{\alpha\beta}$ is the Einstein tensor and $T_{\alpha\beta}$ is the energy-momentum tensor of the fluid. The latter is given by
\[
 T_{\alpha\beta} = (\rho + p) u_\alpha u_\beta + p \, g_{\alpha\beta},
\]
 where $\rho$ denotes the proper energy density, $p$ the pressure and $u^\alpha$ the velocity vector normalized to $u^\alpha u_\alpha = -1$.  The gravitational constant and the speed of light are normalized, i.e., $G=c=1$. The system~\eqref{EE} is underdetermined unless we prescribe a so-called equation of state, $p=p(\rho)$, relating the pressure and proper energy density.
 
\subsection{Spherical symmetry and staticity}
 In the present context, we are primarily interested in static solutions of \eqref{EE}. Such solutions can be viewed as idealized models of stars when they have compact support. In this case, the interior region is described by a perfect fluid and the exterior region is given by an asymptotically flat vacuum spacetime.
 A fundamental result, previously known as the ``fluid ball conjecture'', states that static asymptotically flat spacetimes with perfect fluid sources are spherically symmetric. This conjecture was verified for solutions with positive density $\rho>0$ satisfying $\frac{dp}{d\rho} \geq 0$ by Masood-ul-Alam \cite{MuA}, building upon work of Lindblom and Masood-ul-Alam~\cites{Lin:static,LMuA:static,MuA:static1,MuA:static2} and Beig and Simon~\cites{BS:static1,BS:static2}. It is therefore natural to restrict our attention to not only static but also spherically symmetric solutions of \eqref{EE}. We do, however, not limit our analysis to the standard asymptotically flat situation because it turns out to be a very rigid assumption when dealing with perfect fluids in general relativity. Instead, we also allow solutions with a slower falloff rate and a conical angle at radial infinity.

 Let us recall the setup of \eqref{EE} in the case of spherical symmetry and staticity. The static and spherically symmetric situation amounts to looking at metrics in polar coordinates $(t,r,\theta,\phi)$ of the form
\[
 g = - e^{2\nu(r)} dt^2 + e^{2\lambda(r)} dr^2 + r^2 (d\theta^2 + \sin^2 \theta d\phi^2),
\]
 with unknown metric functions $\nu,\lambda$. From the system \eqref{EE} one obtains that energy momentum conservation is described by the equation
\begin{align}\label{inteq1}
 \frac{d\nu}{dr} = - \frac{dp}{dr} (p+\rho)^{-1}.
\end{align}
 On the other hand, integrating the field equations \eqref{EE} gives $e^{-2\lambda} = 1 - \frac{2m}{r}$, where $m=m(r)$ denotes the mass $m=m(r)$ up to the radius $r$, i.e., 
\begin{align}\label{inteq2}
 m(r) = 4 \pi \int_0^r s^2 \rho(s) \, ds.
\end{align}
 The Einstein--Euler system~\eqref{EE} in spherical symmetry therefore reduces to two coupled nonlinear ordinary differential equations for the mass function $m = m(r)$ and pressure $p=p(\rho(r))$ of the form 
\begin{subequations} \label{staticEE}
\begin{align}
 \frac{dm}{dr} &= 4 \pi r^2 \rho, \label{staticEE1} \\
 \frac{dp}{dr} &= - \frac{\rho m}{r^2} \left( 1 + \frac{p}{\rho} \right)
                  \left( 1 + \frac{4 \pi r^3 p}{m} \right) \left( 1 - \frac{2m}{r} \right)^{-1}. \label{staticEE2}
\end{align}
\end{subequations}
 The second equation \eqref{staticEE2} has been studied extensively and is referred to as the Tolman--Oppenheimer--Volkoff equation. Note that \eqref{staticEE2} is highly nonlinear and singular at the center $r=0$, which largely complicates the analysis of the static system~\eqref{staticEE}. Hardly any solutions in closed form are  known, even for the simplest equations of state. Known analytic solutions with linear equation of state are the flat dust solution, the singular Klein--Tolman solutions~\cite{T} relevant for neutron stars, the Whittaker solution~\cite{W}, a stiff solution by Buchdahl and Land~\cite{BL}, and de Sitter space and the Einstein static universe as solutions with a cosmological constant (see Ivanov \cite{Iva:expl} for a full overview and a new exact solution). The problem of finding explicit solutions is related to the integrability of Abel differential equations of the second kind~\cite{Iva:integrable}.
 Global existence and uniqueness of smooth solutions as functions of $r$ to \eqref{staticEE} for reasonable equations of state  and given central density $\rho_0>0$, on the other hand, were already established in 1991 by Rendall and Schmidt~\cite{RS:static}. Related results in the relativistic and nonrelativistic case have been obtained in \cites{BR:static,BLF:trapped,Mak:local,RR:static,Sch:static}.

 \subsection{(In)finite extent and the role of the equation of state}\label{ssec:extent}
 The global existence and uniqueness result of Rendall and Schmidt~\cite{RS:static}*{Theorem 2} holds for equation of state $\rho=\rho(p)$ which are nonnegative and continuous for $p \geq 0$, and furthermore smooth and satisfy $\frac{d\rho}{dp}>0$ for $p>0$. If the matter has finite extent, then the fluid ball is joined to a (unique) Schwarzschild exterior; hence, the solution is in particular asymptotically flat. If the matter extends to infinity, then $\rho$ tends to zero at infinity. In some borderline cases, the ADM mass of the solution can still be finite (see Remark~\ref{rem:buchdahl} below), but in general it is not. In \cite{RS:static}*{Section 4} some criteria for (in)finite radii are discussed. For example, the finiteness of the integral
\[
 \int_0^{p_0} \frac{dp}{\rho^2(p)} < \infty, \qquad p_0 = p(r=0),
\]
implies that the stellar model has finite extent, a condition that depends on the low-pressure regime only. A similar criterion has been derived by Makino in \cite{Mak:static}*{Theorem 1}. There, a finite radius is tied to the condition
\[
 \frac{\rho}{p} \frac{dp}{d\rho} = \Gamma + O(\rho^{\Gamma-1}), \qquad \textrm{as } \rho \to 0+,~ \Gamma \in (\tfrac{4}{3},2).
\]
On the other hand, a star with finite radius must satisfy
\begin{align}\label{inf}
 \int_0^{p_0} \frac{dp}{\rho(p)+p} < \infty.
\end{align}
 These criteria, however, do not cover all equations of state, and further analysis are often necessary (see, for example, \cites{Hei:barotropic,Sim:barotropic}). In the following, we discuss some important special cases. Whenever the pressure only depends on the density but does not depend on the entropy, the fluid is called barotropic. In order to be able to directly replace the pressure $p$ in \eqref{staticEE} by the energy density $\rho$, we therefore focus on barotropic equations of state.
 
\subsubsection*{Linear equation of state}
 We are particularly interested in the linear equation of state, with sound speed $\sqrt{K}$ normalized to be in $[0,1]$, so that
\begin{align}\label{linearEOS}
 p = K \rho, \qquad 0 < \rho < \rho_0.
\end{align}
 Since
\[
 \int_0^{p_0} \frac{dp}{\rho(p)+p} = \frac{K}{1+K} \int_0^{K\rho_0} \frac{dp}{p} = \infty,
\]
 the (in)finiteness criterion~\eqref{inf} shows that linear equations of state with $K \in (0,1]$ lead to solutions with infinitely extending fluid. The above criteria cannot be applied to piecewise linear equations of state with a hard and a soft phase, i.e., equations of state of the form
 \[
 p(\rho) = \begin{cases}
             0 & \text{if }\rho \leq \rho_0, \\
             K(\rho-\rho_0) & \text{if } \rho > \rho_0,
            \end{cases}
 \]
 where the qualitative behavior changes at a critical density $\rho_0 > 0$. The dynamics of the two-phase model with sound speed $\sqrt{K}=1$ in spherical symmetry, which describes hard stars with a vacuum exterior, has been studied in the work of Christodoulou \cites{Chr1,Chr2,Chr3} and recently by Fournodavlos and Schlue \cite{FS}. 

\subsubsection*{Polytropic equations of state}
 In Newtonian theory, polytropes are given by a power-law equation of state of the form
\begin{align*}
 p = K \rho_N^\frac{n+1}{n},
\end{align*}
 where $\rho_N$ is the Newtonian mass density. For special values of $n$, these polytropes are also adiabates. In the limit $n\to\infty$ we recover the linear equation of state. In general relativity, however, these power-law equations of state are unphysical because the speed of sound could exceed the speed of light (see \cite{JS:PhD}*{p.\ 31f} for a brief discussion on physical equations of state). The corresponding adiabates in general relativity are represented by an equation of state of the form
\[
 p = K \eta^\Gamma,
\]
 where $\eta$ is the rest-mass density and $1<\Gamma <2$ is the (constant) polytropic exponent. The energy density $\rho$ is then of the form
\begin{align}\label{adiabaticEOS}
 \rho = %c^2=1 here
 \eta + \frac{1}{\Gamma-1} p = C p^{\frac{1}{\Gamma}} + \frac{1}{\Gamma - 1} p.
\end{align}
 For $C = K^{-\frac{1}{\Gamma}}$, and $n = \frac{1}{\Gamma -1}$ the corresponding polytropic index, \eqref{adiabaticEOS} is the polytropic equation of state
\begin{align}\label{realpolytropicEOS}
 \rho = K^{\frac{n+1}{n}} p^{\frac{n}{n+1}} + n p.
\end{align}
 For $C=0$ we recover the linear ``gamma-law'' equation of state, i.e.,
 $
   p = K \eta^\Gamma = (\Gamma-1) \rho.
 $
 The polytropic equation of state~\eqref{realpolytropicEOS} was already considered by Tooper~\cite{Top:adia}, who numerically observed instability for $\Gamma \geq \frac{4}{3}$ ($n \geq 3$). From an asymptotic point of view, solutions to \eqref{staticEE} with \eqref{realpolytropicEOS} and power-law polytropic-type equations of state of the form
 \begin{align}\label{polytropicEOS}
 p = K \rho^\frac{n+1}{n}
\end{align}
 essentially behave in the same way because the low-pressure regime dominates (compare, for example, \cites{Top:poly,Top:adia}). Solutions of \eqref{staticEE} with polytropic-type equation of state \eqref{polytropicEOS} with small central densities $\rho_0>0$ and $0<n<5$ also have finite radii and finite masses as observed in \cites{NU:polytropic,RS:static,Sim:barotropic}. However, for $3<n<5$ solutions with infinite extent do occur for larger central densities (compare to the Newtonian case, where finiteness is guaranteed for $n<5$). For $n>5$ the fluid is always unbounded with infinite mass. 

\bigskip
 Despite the frequent use of the linear and polytropic-type equations of state in astrophysics (see, for example, \cites{Cha:linear,FGKL,FT:polytropic,Fla:polytropic,Sh:rot,S:rot}) and evolutionary problems (see, for example, \cites{BLF:trapped,GST,LFS:exterior} for the linear and \cites{BK:lext,Mak:local,O2,R:local} for the polytropic case), the very basic fact that a large class of static solutions and likewise many other solutions are not asymptotically flat has received little attention. In particular, we are not aware of a geometric description that captures the asymptotic behavior of perfect fluids with infinite extent. The main motivation of this paper is to provide such a general geometric framework. We will focus on the linear equation of state \eqref{linearEOS} and the polytropic-type equation of state \eqref{polytropicEOS} with index $n>5$. Although unphysical, the focus on \eqref{polytropicEOS} is natural because it is known to lead to solutions with a similar asymptotic behavior as \eqref{realpolytropicEOS} but is easier to handle analytically. 

\subsection{The asymptotic behavior}\label{ssec:dynamical}
 
 It was already observed by Chandrasekhar~\cite{Chand} in 1972 that spherically symmetric static solutions to \eqref{staticEE} with a linear equation of state \eqref{linearEOS} exhibit an interesting limiting behavior as they approach a singular solution with density function $\rho_\infty(r) = \textrm{constant} \cdot r^{-2}$ as $r \to\infty$. Chandrasekhar computed the asymptotic behavior for $K=\frac{1}{3}$ (and $K=1$) using a reformulation in terms of  Milne variables and observed a spiraling behavior to the singular solution in these coordinates.
 
 In the late 1990s, Makino reformulated \eqref{staticEE} with linear equation of state \eqref{linearEOS} as an autonomous system and used plane dynamical systems theory, more precisely the Poincar\'e--Bendixson Theorem, to obtain that for $K = \frac{1}{3}$ the singular solution is the only element in the $\omega$-limit set and hence all regular solutions converge to it \cite{Mak:static}*{Appendix}. Thus, asymptotically the solutions behave like $\rho(r)\sim \frac{3}{56\pi}r^{-2}$ and $m(r) \sim \frac{3}{14} r$ as $r \to \infty$. He also studied the spiral structure for more general equations of state in \cite{Mak:spiral}, and specifically linear equations of state in \cite{Mak:spiral}*{Section 2}.
 
 Around the same time Heinzle, Nilsson, R\"ohr and Uggla \cites{HU,HRU,NU:linear,NU:polytropic} developed a different dynamical systems approach to study Newtonian as well as relativistic stellar models. Nilsson and Uggla~\cite{NU:polytropic} numerically investigated the asymptotic behavior of solutions with power-law equations of state of the form \eqref{polytropicEOS} and revealed that static solutions with finite extent are the only ones that occur for $n \lesssim 3.339$, but never occur if $n>5$. 
 The more general approach of Heinzle, R\"ohr and Uggla in \cite{HRU} applies to barotropic equations of state that are asymptotically polytropic and linear at the low- and high-pressure regime, respectively. They reformulate the spherically symmetric, static Einstein--Euler system \eqref{staticEE} by introducing certain dimensionless variables to obtain a regular dynamical system on a cube. This reformulation is very well suited for numerical computations and visualization.
 
 While all of the above reformulations as dynamical systems lead to very clear convergence results in the reformulated variables, they cannot be used to derive a convergence \emph{rate\/} in the original formulation. Lower-order terms are crucial to understand the resulting geometric structures and determine their asymptotic behavior. A big drawback is the fact that the radial parameter $r$ is removed in the system \eqref{staticEE} by \emph{implicitly\/} replacing it with a new parameter, for example, in the work of Makino by
 \[
   t(r) = \int_\delta^{r} \frac{ds}{\left(1-\frac{2m(s)}{s}\right)s},
 \]
 whose growth rate with respect to $r$ cannot be controlled well enough a priori. Such implicit reformulations prevent us from interpreting the results obtained in the dynamical systems picture in the original variables, i.e., $m,\rho,p$, and $r$. Nevertheless, the reformulation of \eqref{staticEE} as a dynamical system is also the major analytic tool employed in this paper.

\subsection{Geometric interpretation: Our results}

 In what follows, we provide a geometric description of the asymptotic behavior of solutions to \eqref{staticEE} with linear equation of state~\eqref{linearEOS} and power-law polytropic equation of state~\eqref{polytropicEOS} with index $n>5$. We show that spherically symmetric static perfect fluids with linear equation of state are so-called quasi-asymptotically flat, a concept developed by Nucamendi and Sudarsky~\cite{NS:quasiasymp} which generalizes (and includes) the notion of asymptotic flatness and at the same time admits conformal compactifications. The spatial Riemannian part of the metrics is asymptotically conical.

 \begin{definition}[Quasi-asymptotically flat metrics (AF$\alpha$) \cite{NS:quasiasymp}]\label{quasiflat}
 A spacetime $(M,g)$ with topology $\RR \times (\RR^3 \setminus B_R(0))$, where $B_R(0)$ is a ball of radius $R$ around $0$, is called \emph{quasi-asymptotically flat} (AF$\alpha$) if there exist $\alpha \in (0,1)$ and coordinates $( \tau,\xi,\theta,\phi )$ so that
\begin{align}\label{quasig}
 g = g^\alpha + \tilde g,
\end{align}
 where $g^\alpha$ is the so-called \emph{standard quasi-asymptotically flat metric} (or SAF$\alpha$ metric), given by
 \begin{align}\label{g0}
 g^\alpha = - d\tau^2 + d \xi^2 + (1-\alpha) \xi^2 (d\theta^2 + \sin^2 \theta d\phi^2),
\end{align}
 and $\tilde g$ is of the form
\begin{align*}
 \tilde g_{\mu\nu} dx^\mu dx^\nu &= a_{\tau\tau} d\tau^2 + a_{\xi\xi}d\xi^2+ 2 a_{\xi\tau} d\xi d\tau \\
&\quad + \xi^2 [ a_{\theta\theta} d\theta^2 + a_{\phi\phi} \sin^2 \theta d\phi^2 + 2 a_{\theta\phi} \sin \theta d\theta d\phi ] \\
&\quad + 2\xi [ a_{\tau\theta} d\tau d\theta + a_{\xi\theta} d\xi d\theta ] + 2\xi [ a_{\tau\phi} \sin\theta d\tau d\phi + a_{\xi\phi} \sin\theta d\xi d\phi ],
\end{align*}
 with $a_{\mu\nu} = o(\xi^{-\frac{1}{2}})$ as $\xi \to \infty$.
\end{definition}

 Note that the SAF$\alpha$ metrics $g^\alpha$ play the same role as the Minkowski metric does for asymptotically flat spacetimes. 
 
 Our result is formulated in this framework of quasi-asymptotic flatness.
 
\begin{theorem}[Linear equation of state]\label{mainthm1}
 The unique global smooth solution to the initial value problem of the static Einstein--Euler equations~\eqref{staticEE} in spherical symmetric with linear equation of state $p = K \rho$, for fixed $K \in (0,1)$, and central density $\rho_0>0$, is quasi-asymptotically flat.
 
 More precisely, in coordinates $(\tau = r^{\frac{2K}{1+K}}t,\xi=\frac{r}{\sqrt{1-\alpha}},\theta,\phi)$, the solution converges to the standard quasi-asymptotically flat metric
 \[
  g^\alpha = -d\tau^2 + d\xi^2 + (1-\alpha) \xi^2 (d\theta^2 + \sin^2 \theta d\phi^2),
 \]
 with $\alpha = \frac{4K}{(1+K)^2+4K}$, with rate $o(\xi^{-\frac{1}{2}})$ as $\xi\to\infty$.
\end{theorem}

\medskip
 The asymptotic behavior is different for solutions to \eqref{staticEE} with power-law polytropic-type equation of state \eqref{polytropicEOS}. If $n<3$ only solutions with finite extent occur, if $n>5$ only solutions with infinite extent occur \cite{NU:polytropic}. Although the latter solutions represent metrics that converge to the flat spacetime, they are neither asymptotically nor quasi-asymptotically flat in the strict sense due to a slower convergence rate than $o(r^{-\frac{1}{2}})$ and an infinite ADM mass. This deviation from the standard (quasi-)asymptotically flat situation is captured in our notion of scaled quasi-asymptotic flatness.

 \begin{definition}[Scaled quasi-asymptotically flat metrics (AF$\alpha\beta$)]\label{def:sqaf}
 A spacetime $(M,g)$ with topology $\RR \times (\RR^3 \setminus B_R(0))$, where $B_R(0)$ is a ball of radius $R$ around $0$, is called \emph{scaled quasi-asymptotically flat\/} (AF$\alpha\beta$), if there exist $\alpha \in [0,1)$, $\beta > 0$ and coordinates $(\tau,\xi,\theta,\phi)$ so that
\begin{align}\label{gabasymp}
 g = \gab + \tilde g,
\end{align}
 where $\gab$ is the so-called \emph{standard scaled quasi-asymptotically flat metric} (or SAF$\alpha\beta$ metric), given by
 \begin{align}\label{gab}
 \gab = -d\tau^2 + \xi^{2\beta} \left( d\xi^2 + (1-\alpha) \xi^2 (d\theta^2 + \sin^2 \theta d\phi^2) \right),
 \end{align} 
 and $\tilde g$ is of the form
\begin{align*}
 \tilde g_{\mu\nu} dx^\mu dx^\nu &= a_{\tau\tau} d\tau^2 + a_{\xi\xi} \xi^{2\beta} d\xi^2+ 2 \xi^\beta a_{\xi\tau} d\xi d\tau \\
&\quad + \xi^{2(1+\beta)} [ a_{\theta\theta} d\theta^2 + a_{\phi\phi} \sin^2 \theta d\phi^2 + 2 a_{\theta\phi} \sin \theta d\theta d\phi ] \\
&\quad + 2\xi^{1+\beta} [ a_{\tau\theta} d\tau d\theta + a_{\tau\phi} \sin\theta d\tau d\phi ] + 2\xi^{1+2\beta} [ a_{\xi\theta} d\xi d\theta + a_{\xi\phi} \sin\theta d\xi d\phi ],
\end{align*}
with $a_{\mu\nu} = o(\xi^{-\frac{1}{2}})$ as $\xi \to \infty$.
\end{definition}

\begin{remark}[Relation to asymptotic flatness]\label{rem:af}
 If $\alpha=\beta=0$, then $g^{(\alpha,\beta)} = g^\alpha = g^0$ is the flat Minkowski metric. In fact, we can write the Minkowski metric as $g^{(\alpha,\beta)}$ for \emph{any} choice of $\beta\geq 0$ and fixed $1-\alpha = \frac{1}{(1+\beta)^2}$ if we choose a new radial coordinate $\xi = \sqrt[1+\beta]{(1+\beta)r}$, since then $dr = \xi^ \beta d\xi$. Thus while AF$\alpha\beta$ metrics with $1-\alpha = \frac{1}{(1+\beta)^2}$ are asymptotic to the flat spacetime, their convergence rate is generally too slow to fall into the standard asymptotically flat regime (the ADM mass is infinite). If, on the other hand, $1-\alpha < \frac{1}{(1+\beta)^2}$ the corresponding AF$\alpha\beta$ metrics converge (slowly) to a spacetime metric $g^\alpha$ with a true conical angle.
 
 Alternatively, we could have defined the SAF$\alpha\beta$ metrics in \eqref{gab} to be of the form
 \[
  \gab = -d\tau^2 + (1+\beta) \zeta^{2\beta} \left( d\zeta^2 + (1-\alpha) \zeta^2 (d\theta^2 + \sin^2 \theta d\phi^2) \right),
 \]
 to emphasize the transformation of the Minkowski metric via $\zeta = \sqrt[1+\beta]{\sqrt{1+\beta}r}$ for any $\beta \geq 0$ and fixed $\alpha = 0$ (independent of $\beta$). For $\alpha > 0$ we would obtain the SAF$\alpha$ metrics in \eqref{g0} with a true conical angle.
\end{remark}
 
 We formulate our result in this scaled quasi-asymptotically flat setting of Definition~\ref{def:sqaf}.
 
\begin{theorem}[Polytropic equation of state]\label{mainthm2}
 The unique global smooth solution to the initial value problem of the static Einstein--Euler equations~\eqref{staticEE} in spherical symmetry with power-law polytropic equation of state $p = K \rho^{\frac{n+1}{n}}$, for $K \in (0,1)$ and $n>5$ fixed, and central density $\rho_0>0$, is scaled quasi-asymptotically flat.
 More precisely, in coordinates $(\tau=e^{\nu(r)}t,\xi=\sqrt[1+\beta]{(1+\beta)r},\theta,\phi)$, the solution is asymptotic to 
 \[
  \gab = -d\tau^2 + \xi^{2\beta} \left( d\xi^2 + (1-\alpha) \xi^2 (d\theta^2 + \sin^2 \theta d\phi^2) \right),
 \] 
 where $\beta > \frac{n-5}{4}$ and $1-\alpha = \frac{1}{(1+\beta)^2}<\frac{16}{(n-1)^2}$,
 with rate $o(\xi^{-\frac{1}{2}})$ as $\xi\to\infty$.
 
 In the original coordinates $( t, r, \theta, \phi)$, the spatial part of these solutions is asymptotic to the Euclidean metric
 \[
  h^0 = dr^2 + r^2 (d\theta^2 + \sin^2 \theta d\phi^2),
 \]
 with convergence rate $O(r^{-\frac{2}{n-1}})$ as $r \to\infty$.
\end{theorem}

\begin{remark}
 In view of the relation of $g^{(\alpha,\beta)}$ to the Minkowski metric when $1-\alpha = \frac{1}{(1+\beta)^2}$, discussed in Remark~\ref{rem:af}, Theorem~\ref{mainthm2} shows that solutions to \eqref{staticEE} with power-law polytropic equation of state are in fact asymptotic to the flat metric. However, the convergence rate of $O(r^{-\frac{2}{n-1}})$ as $r \to\infty$ is too slow to interpret this behavior in the standard asymptotically flat setting which requires $o(r^ {-\frac{1}{2}})$. It is conceivable that for related equations of state, e.g., equations of state that are asymptotically linear/polytropic in the low pressure regime, also a nontrival conical angle would occur, expressed by an inequality $1-\alpha < \frac{1}{(1+\beta)^2}$.
\end{remark}

\begin{remark}[The borderline case]\label{rem:buchdahl}
 The asymptotic behavior of solutions to \eqref{staticEE} for equations of state that become polytropic of index $n=5$ at the low pressure regime (recall that the low pressure regime is critical for the behavior at spatial infinity) is already known. It has been shown that the so-called Buchdahl equation of state \cite{Buchdahl}, given by
\[
 p = \frac{1}{6} \frac{\rho^{6/5}}{\rho_0^{1/5} - \rho^{1/5}},  \qquad 0 < \rho < \rho_0,
\]
 and generalizations thereof \cite{BK} yield asymptotically flat solutions with a fluid extending to infinity. This essentially agrees with our observation in Theorem~\ref{mainthm2} if we would consider the limit $n\to 5$, because the necessary falloff rate to obtain an asymptotically flat spacetime requires $o(r^{-\frac{1}{2}})$ \cite{Chr:mass}. As such the index $n=5$ is the borderline case between finite and infinite mass/extent.
\end{remark}

In what follows, we briefly mention some of the properties of our geometric framework and discuss our results in the wider context of relativistic perfect fluid models and of Einstein--matter equations in general. For further definitions, properties and discussions related to the (scaled) quasi-asymptotically flat metrics, we refer the reader to Section~\ref{sec2} of this paper.

 \begin{remark}[Generalized ADM mass]
 Since static perfect fluids with linear equations of state and polytropic equations of state with index $n>5$ are not asymptotically flat, their ADM masses are infinite. In the framework of quasi-asymptotic flatness, however, one can substitute the infinite ADM mass by the use of a so-called ADM$\alpha$ mass introduced by Nucamendi and Sudarsky~\cite{NS:quasiasymp}. This notion of mass  coincides with the monopole mass used in \cite{BV:monopole}. The standard quasi-asymptotically flat metric $g^\alpha$ has vanishing ADM$\alpha$ mass. The ADM$\alpha$ mass of a regular solutions described in Theorem~\ref{mainthm1}, however, remains unknown and we argue in Remark~\ref{rem:les} that it could be unbounded below. Hence also the concept of the ADM$\alpha$ mass is of little use in the analysis of perfect fluids.
 
 Based on our notion of scaled quasi-asymptotic flatness with reference metrics $g^{(\alpha,\beta)}$ as in Theorem~\ref{mainthm2} we consider a na\"{\i}ve definition of an ADM$\alpha\beta$ mass in Remark~\ref{rem:ADMalphabeta}. For reference Riemannian metrics
 \[
  h^{(\alpha,\beta)} = \xi^{2\beta} \left(d\xi^2 + (1-\alpha) \xi^2 (d\theta^2 + \sin^2\theta d\varphi^2)\right)
 \]
 and $h$ a scaled quasi-asymptotically flat metric, we let
 \begin{align}\label{mnew}
m_{\mathrm{ADM}\alpha\beta} (h) = \frac{1}{16 \pi (1-\alpha)} \lim_{\xi \to \infty} \xi^{-\beta} \int_{S_\xi(0)}
(h^{(\alpha,\beta) ik} h^{(\alpha,\beta) jl} - h^{(\alpha,\beta) ij} h^{(\alpha,\beta) kl}) \nabla^{(\alpha,\beta)}_j (h_{kl}) \, dS_i,
 \end{align}
 where $dS_i$ is the $i$-th surface element and $\nabla^{(\alpha,\beta)}$ is the covariant derivative with respect to $h^{(\alpha,\beta)}$.
 We will see that, if $\alpha=\beta=0$, then \eqref{mnew} is just the ADM mass, i.e.,
 \[ m_{\mathrm{ADM}00}(h) = m_{\mathrm{ADM}}(h) .\]
 The advantage of \eqref{mnew}, however, is that it makes sense also for metrics that are asymptotically flat in a nonstandard sense, namely for asymptotically conical metrics and those with a slow converge rate as described in Definition~ \ref{def:sqaf}. 
 The slower convergence rate is accounted for by multiplication by $\xi^{-\beta}$ and the deficit angle is accounted for by dividing by $1-\alpha$.
 
  Since the solutions studied in Theorems~\ref{mainthm1} and \ref{mainthm2} are of AF$\alpha$ and AF$\alpha\beta$ form, we obtain on a hypersurface $\Sigma_\tau$ ($\tau$ is a rescaled time variable)
\[
  m_{\text{ADM}\alpha\beta}(h^\text{lin}|_{\Sigma_\tau}) = O(1) + \tau^2 O(1) < \infty, 
\]
 for $\beta  > \frac{1-K}{1+3K}$ and $1-\alpha < \frac{(1+3K)^2}{4((1+K)^2+4K)}$, and
\[
 m_{\text{ADM}\alpha\beta}(h^\text{poly}|_{\Sigma_\tau}) = \sqrt[n-1]{\frac{K^{n}}{2\pi} \frac{(n-3)(n+1)^n}{(n-1)^{2}}} \in (0,\infty), 
\]
 for $\beta=\frac{n-3}{2}$ and $1-\alpha = \frac{4}{(n-1)^2}$.
 For further details, see Remarks~\ref{rem:ADMablinear} and \ref{rem:ADMabpoly}.
 It remains to be checked whether any such notion of ADM$\alpha\beta$ mass can be derived in a coherent and coordinate-invariant fashion, and if and in what sense such a mass  could be preserved in time. A rigorous approach could be based on related work on other masses \cites{Chr:mass,M:mass,NS:quasiasymp}.
\end{remark}

\begin{remark}[Dynamics]
 Knowledge about the asymptotic behavior of static perfect fluids is also of use in the full dynamic picture when constructing local solutions out of initial data sets. For example, recent results of LeFloch and the second author~\cite{BLF:trapped} on the formation of trapped surfaces make use of initial data that are constructed as large focused perturbations of static spherically symmetric perfect fluids with linear equation of state.
 Local existence results for solutions to the Einstein--Euler equations \eqref{EE} with, in particular, power-law polytropic equations of state \eqref{polytropicEOS} but compact support have been studied in the smooth case by Rendall~\cite{R:local}.
 Using initial data with compact support or satisfying certain falloff conditions, Brauer and Karp~\cites{BK:lclass,BK:well,BK:lext} constructed solutions in a class of weighted Sobolev spaces with fractional order depending on the polytropic exponent $\Gamma$. However, already Makino~\cite{Mak:local,Mak:NM} remarked that general static solutions of \eqref{staticEE} are actually excluded from the class of density distributions allowed in the setting of Brauer and Karp, and proves existence of smooth solutions near an equilibrium. Oliynyk~\cites{O1,O2} recently also obtained local existence results in the realistic case of compact barotropic fluid bodies with a free matter-vacuum boundary.
 For initial value formulations that do admit smooth static solutions with infinite extent studied in Theorems~\ref{mainthm1} and \ref{mainthm2}, we expect that our geometric interpretation applies to other solutions studied in those frameworks as well.
 In fact, problems with the common geometric paradigm of asymptotic flatness already occur when one wants to consider rotating stars with a vacuum exterior. These stars are modeled by stationary, axisymmetric perfect fluid spacetimes and one would expect that they are---in analogy to the static case---glued to a Kerr vacuum exterior. This is surprisingly \emph{not\/} the case \cites{CMMR:rotating,MM:rot}, but if a rotating star collapses to a black hole, it is expected that the exterior region is approximately Kerr \cites{BHMLRSFS,FS:rot,N:rot,PS:rot,Sh:rot,S:rot}.
\end{remark}

\begin{remark}[Other matter fields]
 The prototype for quasi-asymptotically flat metrics are the global monopole spacetimes studied by Barriola and Vilenkin \cite{BV:monopole}, Nucamendi and Sudarsky \cite{NS:blackholes} and others. Conical singularities (albeit in the center) also occur in electromagnetic fields, more precisely in asymptotically flat spherically symmetric static solutions of the Einstein--Maxwell equations as shown by Tahvildar-Zadeh~\cite{TZ:electrostatic}. 
 Spherically symmetric static solutions to the Einstein--Vlasov equations can also have infinite nonasymptotically flat extent, and criteria for collisionless gas related to those in the perfect fluid case which guarantee solutions with compact support have been discussed by Andr\'{e}asson, Fajman, Ramming, Rein and Thaller \cites{AFT:vlasov,AR:static,RR:static}. Moreover, in a dynamical collapse scenario Rendall and Vel\'{a}zquez \cite{RV:veiled} obtained solutions to the Einstein--Vlasov equations with naked-type singularities that are selfsimilar and not asymptotically flat. Overall it is apparent that spacetimes with matter extending to infinity that are not asymptotically flat are not merely an artifact of these theories but in fact a common feature in general relativity worth exploring. After all, asymptotic flatness is an idealization that may simply not be suitable for many mathematical and physical scenarios. 
 The very basic vacuum solutions with positive and negative cosmological constant, de Sitter and anti-de Sitter spacetimes, respectively, are prominent examples of asymptotically simple manifolds (in the sense of Penrose~\cites{P:asymp,P:asymp0,P:asymp1,PR2}) that are not asymptotically flat. Spacetimes with an asymptotically  hyperbolic (anti-de Sitter) behavior, in particular, became increasingly important in the last few years \cites{A:AdS,AM:adS,CD:hyp,DGS:hyp,W:hyp}.
 We believe that, along these lines, the geometric notion of scaled quasi-asymptotic flatness can be verified and adopted in several other scenarios in general relativity as well.
\end{remark}

{\bf Outline.}
 This paper is structured as follows. Section~\ref{sec2} builds the geometric core of this paper. The concept of (scaled) quasi-asymptotic flatness is described in detail and related to the concept of asymptotic simplicity, i.e., conformal compactifications at null infinity. Furthermore, we recall the notions of ADM mass and ADM$\alpha$ mass and extend it to include spacetimes that are scaled quasi-asymptotically flat. Some simplifications for the spherically symmetric setting are also derived, which will be of use later. In Section~\ref{sec3} we see that solutions to \eqref{staticEE} with linear equation of state have infinite ADM mass but converge to a standard quasi-asymptotically flat singular solution with vanishing ADM$\alpha$ mass. This proves Theorem~\ref{mainthm1}. The analytical tool used here is the reformulation of \eqref{staticEE} as a dynamical system and a stability analysis via linearization. A similar but slightly more involved procedure is applied in Section~\ref{sec4} to analyze solutions with polytropic equations of state. This analysis and a geometric reformulation yield Theorem~\ref{mainthm2}.

\medskip
{\bf Notations and conventions.} Throughout the manuscript, we use Greek indices $\mu,\nu$ etc.\ to denote the components $0,1,2,3$ of a spacetime metric $g$, and Latin indices $i,j,k$ etc.\ to denote the components $1,2,3$ of the spatial metric (often $h$). The signature of $g$ is $(-,+,+,+)$. We use the Einstein summation convention.

%%%%%%%%%%%%%%%%%%%%%%%%%%%%%%%%%%%%%%%%%%%%%%%%%%%%%%%%%%%%%%%%%%%%%%%%%%%%%%%%%%%%%%%%
%%%%%% SECTION 2 %%%%%%%%%%%%%%%%%%%%%%%%%%%%%%%%%%%%%%%%%%%%%%%%%%%%%%%%%%%%%%%%%%%%%%%
%%%%%%%%%%%%%%%%%%%%%%%%%%%%%%%%%%%%%%%%%%%%%%%%%%%%%%%%%%%%%%%%%%%%%%%%%%%%%%%%%%%%%%%%

\section{Beyond asymptotic flatness}
\label{sec2}

\subsection{Asymptotically and (scaled) quasi-asymptotically flat metrics} \label{ssec:metric}
 We are primarily interested in spherically symmetric metrics. For polar coordinates $(t,r,\theta,\phi)$ we can write the metric tensor in the form
\begin{align}\label{metric}
 g = - e^{2\nu(r)} dt^2 + e^{2\lambda(r)} dr^2 + r^2 (d\theta^2 + \sin^2 \theta d\phi^2),
\end{align}
 where $\nu$ and $\lambda$ are the unknown metric variables. % (see \cite{HE}*{Appendix B}).
 Asymptotic flatness is tied to the limiting behavior (with specific decay rates)
\begin{align}\label{flat}
 \lim_{r\to\infty} \nu(r) = \lim_{r\to\infty} \lambda(r) = 0,
\end{align}
 which only holds for a very limited number of equations of state. In general, we will not observe that $\lambda(r)$ tends to $0$ at infinity but to some positive value $\Lambda$ such that
\begin{align}\label{almostflat}
 \lim_{r\to\infty} e^{2\lambda(r)} = e^{2\Lambda} > 1.
\end{align}
 For example, in the specific situation of global monopole spacetimes the asymptotic behavior
\[
 e^{2\nu(r)} = e^{-2\lambda(r)}=1-\alpha-\frac{2M}{r} + O(r^{-2})
\]
 has been studied by Barriola and Vilenkin~\cite{BV:monopole} and later led Nucamendi and Sudarsky~\cite{NS:quasiasymp} to introduce the concept of quasi-asymptotic flatness introduced in Definition~\ref{quasiflat}.

Metrics that are quasi-asymptotically flat are asymptotic to metrics of the form
\[
  g^\alpha = -d\tau^2 + d\xi^2 + (1-\alpha) \xi^2 (d\theta^2 + \sin^2 \theta d\phi^2),
\]
for some $\alpha \in (0,1)$, as $\xi\to \infty$.

These metrics $g^\alpha$ play the same role as the Minkowski metric does for asymptotically flat spacetimes. Note that we allow the slightly weaker falloff condition $o(\xi^{-\frac{1}{2}})$ rather than $O(\xi^{-1})$ which was used by Nucamendi and Sudarsky in \cite{NS:quasiasymp}. This is in accordance with the asymptotically flat situation (see, for example, \cite{Chr:mass}) and the definition of a mass in Section~\ref{ssec:adma}.

Due to the occurrence of even slower convergence rates $o(r^{-\frac{1}{2(1+\beta)}})$, for some $\beta > 0$, in our analysis of solutions to the Einstein--Euler equations \eqref{EE}, we further introduced the concept of scaled quasi-asymptotic flatness in Definition~\ref{def:sqaf}. The basic idea is to study metrics that are asymptotic to those of the form
\[
 \gab = -dt^2 + (1+\beta) r^{2\beta} \left( dr^2 + (1-\alpha) r^2 (d\theta^2 + \sin^2 \theta d\phi^2) \right),
\]
for $\alpha \geq 0$ and $\beta \geq 0$ with convergence rate $o(r^{-\frac{1}{2}})$ as $r \to \infty$.

In the Sections to come, we will review general properties of (scaled) quasi-asymptotically flat metric, such as the existence of conformal compactifications, the spherically symmetric situation and notions of masses. Since $g^\alpha = g^{(\alpha,0)}$ we will only consider the general case of scaled quasi-asymptotically flat metrics, and remark on specific results if $\beta =0$ separately.

\subsection{Asymptotic simplicity}

 In \cite[Section 3]{NS:quasiasymp} it was shown that the SAF$\alpha$ spacetime $(M,g^\alpha)$ can be conformally compactified in the sense of Penrose \cites{P:asymp,P:asymp0,P:asymp1} and is therefore asymptotically simple. More precisely, the concepts of future null infinity $\mathscr{I}^+$ and past null infinity $\mathscr{I}^-$ exist, but not spatial infinity $\iota^0$. We recall the precise definition of (weakly) asymptotically simple spacetimes (see, for example, \cite{PR2}*{Section 9.6}) and provide a corresponding proof for the SAF$\alpha\beta$ spacetime $(\RR^{1+3},\gab)$.

\begin{definition}[(Weakly) asymptotically simple spacetimes]\label{def:simple}
  A spacetime $(M,g)$ is called \emph{asymptotically simple\/} if there exists a smooth spacetime $(\widehat M,\hat g)$ with boundary such that
\begin{enumerate}
 \item $M$ is the interior of $\widehat M$, and hence $\widehat M= M \cup \mathscr{I}$ with $\partial \widehat M = \mathscr{I}$,
 \item the unphysical metric $\hat g$ is conformal to the physical metric $g$, i.e., there exists a smooth conformal factor  $\Omega$ on $\widehat M$ such that
\begin{itemize}
  \item $\hat g_{\mu\nu} = \Omega^2 g_{\mu\nu}$ in $M$
  \item $\Omega > 0$ on $M$, and $\Omega = 0$, $\nabla_\mu \Omega \neq 0$ along $\mathscr{I}$,
\end{itemize}
 \item every inextendible null geodesic in $M$ has a future and past end point on $\mathscr{I}$.
\end{enumerate}
 A spacetime $(M,g)$ is called \emph{weakly asymptotically simple\/} if there exists an asymptotically simple $N$ such that for a neighborhood $U$ of $\mathscr{I}$ in $\widehat N$, $U \cap N$ is isometric to a subset of $M$.
\end{definition}

 Note that (iii) requires that the spacetime is null geodesically complete and hence rules out singularities, black holes, etc. Weakly asymptotically simple spacetimes, however, may possess further ``infinities''.

 In the sense of Penrose, a spacetime is asymptotically flat if it is weakly asymptotically simple and \emph{asymptotically empty\/}, i.e., the Ricci tensor vanishes in a neighborhood of $\mathscr{I}$.

\begin{proposition}\label{thm:as}
 The SAF$\alpha\beta$ spacetime $(M,\gab)$ is asymptotically simple for any $\alpha \in [0,1)$ and $\beta \geq 0$. It is (asymptotically) empty if and only if $\alpha = 0$.
\end{proposition}

\begin{proof}
 We use standard conformal compactification of Minkowski space, to show that $(M,\gab)$ is asymptotically simple. The transformation
\[
 u=t+\frac{r^{1+\beta}}{1+\beta}, \qquad v = t-\frac{r^{1+\beta}}{1+\beta},
\]
immediately implies that $-dt^2+r^{2\beta}dr^2 = - du dv$ and $\frac{r^{2(1+\beta)}}{(1+\beta)^2} =  \frac{(u-v)^2}{4}$. Next we compactify $u$ and $v$ by choosing
\[
 T = \arctan v + \arctan u, \qquad R = \arctan v - \arctan u.
\]
 Therefore, 
\[
 u = \tan \left( \frac{T-R}{2} \right), \qquad v = \tan \left( \frac{T+R}{2} \right),
\]
 The range of $T$ and $R$ is $T+R,T-R \in (-\pi,\pi)$, $R \in (0,\pi)$ and can be extended to include future and past null infinity, i.e., $T+R,T-R = \pm \pi$. The fact that $-du dv = \frac{(1+u^2)(1+v^2)}{4} ( - dT^2 + dR^2 )$ suggests that we should use the conformal factor
\[
 \Omega^2 := \frac{4}{(1+u^2)(1+v^2)} = 4 \cos^2 \left(\frac{T-R}{2}\right) \cos^2 \left(\frac{T+R}{2}\right).
\]
 We verify that $\Omega$ satisfies all conditions of Definition~\ref{def:simple}. It is clear that $\Omega^2 > 0$ on $M$ given by  $T+R,T-R \in (-\pi,\pi)$, $R \in (0,\pi)$ and $\Omega = 0$ for $T+R,T-R = \pm \pi$. Furthermore,
\[
 \nabla_T \Omega = -2 (\cos T + \cos R) \sin T, \qquad \nabla_R \Omega = -2 (\cos T + \cos R) \sin R,
\]
 which do not vanish at the boundary because $R \neq 0, \pi$.
 The transformed SAF$\alpha\beta$ metric $\gab$ \eqref{gab} reads
\begin{align*}
 \hat g^{(\alpha,\beta)} := \Omega^2 \gab &= -dT^2 + dR^2 + \frac{(u-v)^2}{(1+u^2)(1+v^2)} (1+\beta)^2 (1-\alpha) (d\theta^2 + \sin^2 \theta d\phi^2) \\
                          &= -dT^2 + dR^2 + \sin^2R \, (1+\beta)^2 (1-\alpha) (d\theta^2 + \sin^2 \theta d\phi^2).
\end{align*}

 Since the only nonvanishing terms of the Ricci curvature tensor are
 \[
  R_{22} = (1-(1-\alpha)(1+\beta)^2), \qquad R_{33} = (1-(1-\alpha)(1+\beta)^2) \sin^2 \theta,
 \]
 the metric $\gab$ is in general \emph{not} asymptotically empty. It is asymptotically empty if and only if $\alpha = \beta = 0$, i.e., if $\gab$ is the Minkowski metric.
\end{proof}

\begin{remark}
 In general, we do not expect that (scaled) quasi-asymptotically flat spacetimes as described in Definition~\ref{def:sqaf} admit a \emph{smooth\/} conformal compactification. We can, however, show that the prescribed decay rate yields a continuous conformal compactification and expect that further restrictions on the decay rate of the derivatives yield more regular conformal compactifications in accordance with the asymptotically flat situation (for a discussion in the latter framework see, for example, \cite{F:lrr}*{Section 2.3} and \cite{F:nullinf}*{Section 3}). To this end one has to utilize the same embedding and unphysical spacetime
 \[ \widehat M = \{ (T,R) \, | \, R\in (0,\pi), T\pm R \in (-\pi,\pi) \} \cup \{ T+R = \pm \pi \} \cup \{T-R = \pm \pi \}, \]
 and the same conformal factor,
 \[ \Omega = 2 \cos \left( \frac{T-R}{2} \right) \cos \left( \frac{T+R}{2} \right), \]
 as in the proof of Proposition~\ref{thm:as}.
\end{remark}

\subsection{(Scaled) quasi-asymptotic flatness for static spherically symmetric spacetimes}

 We show how the asymptotic behavior of spherically symmetric metrics \eqref{metric} that are not asymptotically flat can be analyzed in the setting of scaled quasi-asymptotic flatness, depending on the limiting behavior of $\lambda$ and $\nu$ as $r\to\infty$.

 In Sections~\ref{sec3} and \ref{sec4} we verify that perfect fluids with linear and polytropic equation of state (for $n>5$) satisfy these conditions.

\begin{proposition}\label{prop:sqaf}
 Suppose $g$ is a static spherically symmetric Lorentzian metric of the form~\eqref{metric}, i.e., in local coordinates $(t,r,\theta,\phi)$ we can write
\[
 g = - e^{2\nu(r)} dt^2 + e^{2\lambda(r)} dr^2 + r^2 (d\theta^2 + \sin^2 \theta d\phi^2).
\]
 If for some $\Lambda \geq 0$ and $\beta \geq 0$ the functions $\lambda$ and $\nu$ satisfy\footnote{Thus, in particular, $\lim_{r\to\infty}\lambda(r) = \Lambda$ but $\nu(r)$ may diverge.}
\begin{equation}\label{ss3}
\begin{split}
 \nu'(r) &= o(r^{-\frac{1}{2(1+\beta)}}), \\ 
 e^{2\lambda(r)-2\Lambda} -1 &= o(r^{-\frac{1}{2(1+\beta)}}) ,
\end{split}
\end{equation}
as $r\to\infty$, then there exists $\alpha \in [0,1)$ and a coordinate system $(\tau,\xi,\theta,\phi)$ such that $g$ is of the form
\[
 g = g^{(\alpha,\beta)} + \tilde g,
\]
with decay rates $a_{\mu\nu} = o(\xi^{-\frac{1}{2}})$ as $\xi\to\infty$ for the scaled components of $\tilde g$. In particular, $g$ is scaled quasi-asymptotically flat (AF$\alpha\beta$) in the sense of Definition~\ref{def:sqaf} with scaling exponent $\beta$ and deficit angle \[(1-\alpha)\pi = (1+\beta)^{-2} e^{-2\Lambda} \pi.\]
\end{proposition}

\begin{proof}
 Set $\tau := e^{\nu(r)}t$ and $\xi := \sqrt[1+\beta]{(1+\beta)e^\Lambda r}$. Since $\beta \geq 0$,
\begin{align*}
 d\tau &= \nu'(r) \tau dr + e^{\nu(r)} dt , \\
 d\xi &= \frac{1}{1+\beta} \left( (1+\beta) e^\Lambda r \right)^{\frac{1}{1+\beta}-1} (1+\beta) e^\Lambda dr = e^{\Lambda} \xi^{-\beta} dr.
\end{align*}
Hence,
\begin{align*}
 e^{2\nu(r)} dt^2 &= (d\tau - \nu'(r) \tau dr)^2 \\
 &= (d\tau - \nu'((1+\beta)^{-1} e^{-\Lambda} \xi^{1+\beta}) e^{-\Lambda} \tau \xi^{\beta} d\xi)^2 \\
 %&= d\tau^2 - 2 \nu'(r) \tau d\tau dr + (\nu'(r)\tau)^2 dr^2 \\
 &= d\tau^2 - 2 e^{-\Lambda} \nu'((1+\beta)^{-1} e^{-\Lambda} \xi^{1+\beta}) \tau \xi^{\beta} d\tau d\xi \\
 & \quad + e^{-2\Lambda} \nu'((1+\beta)^{-1} e^{-\Lambda} \xi^{1+\beta})^2 \tau^2 \xi^{2\beta} d\xi^2,
\end{align*}
and
\begin{align*}
 e^{2\lambda(r)} dr^2 &= e^{2\lambda(r)-2\Lambda} e^{2\Lambda} dr^2
 = e^{2\lambda((1+\beta)^{-1} e^{-\Lambda} \xi^{1+\beta}) - 2\Lambda} \xi^{2\beta} d\xi^2.
\end{align*}
 The metric of the unit sphere, i.e., $d\Omega^2 = d\theta^2 + \sin^2 \theta d\phi^2$, remains unchanged and we thus have that
\begin{align*}
 g = & - e^{2\nu(r)} dt^2 + e^{2\lambda(r)} dr^2 + r^2 d\Omega^2 \\
   = & - d\tau^2 + \xi^{2\beta} d\xi^2 + (1+\beta)^{-2} e^{-2\Lambda} \xi^{2(1+\beta)} d\Omega^2 \\
 & + 2 e^{-\Lambda} \nu'((1+\beta)^{-1} e^{-\Lambda} \xi^{1+\beta}) \tau \xi^{\beta} d\tau d\xi \\
 & +\left[ - 1 + e^{2\lambda((1+\beta)^{-1} e^{-\Lambda} \xi^{1+\beta}) - 2\Lambda} - e^{-2\Lambda} \nu'((1+\beta)^{-1} e^{-\Lambda} \xi^{1+\beta})^2 \tau^2 \right] \xi^{2\beta} d\xi^2.
\end{align*}
 For $\alpha \in (0,1)$ defined by $1-\alpha = (1+\beta)^{-2} e^{-2\Lambda}$, $g$ is therefore of the form
 \[
  g = g^{(\alpha,\beta)} + \tilde g,
 \]
 and it remains to verify the decay rates for $\tilde g$.
The nonzero components of $\tilde g$, that is $a_{\tau\xi}$ and $a_{\xi\xi}$ as in Definition~\ref{def:sqaf}, satisfy
\begin{align*}
a_{\tau\xi} &= 2 e^{-\Lambda} \nu'((1+\beta)^{-1} e^{-\Lambda} \xi^{1+\beta}) \tau = o(\xi^{-\frac{1}{2}}),
\end{align*}
and
\begin{align*}
 a_{\xi\xi} &= - 1 + e^{2\lambda((1+\beta)^{-1} e^{-\Lambda} \xi^{1+\beta}) - 2\Lambda} - e^{-2\Lambda} \nu'((1+\beta)^{-1} e^{-\Lambda} \xi^{1+\beta})^2 \tau^2 \\
                     &=  o(\xi^{-\frac{1}{2}}) - o (\xi^{-1}) = o(\xi^{-\frac{1}{2}}),
\end{align*}
due to the assumptions~\eqref{ss3}. This verifies the conditions of Definition~\ref{def:sqaf}.
\end{proof}

\begin{remark}[Higher decay rates]\label{higher}
 Suppose $g$ is a spherically symmetric metric with better decay rates, i.e., for a $\beta > -1$,
 \begin{equation}\label{ss4}
\begin{split}
 \nu'(r) &= O(r^{-\frac{1}{1+\beta}}), \\
 e^{2\lambda(r)-2\Lambda} -1 &= O(r^{-\frac{1}{1+\beta}}) ,
\end{split}
\end{equation}
 as $r \to \infty$, then $g$ of course also satisfies \eqref{ss3} since $-\frac{1}{2(1+\beta)} > -\frac{1}{1+\beta}$. Hence by Proposition~\eqref{prop:sqaf}, $g$ is AF$\alpha\beta$. However, the components $a_{\mu\nu}$ then satisfy a better decay rate $o(\xi^{-1})$ then if $\beta$ would have been chosen optimally. This will be useful later in the context of an ADM$\alpha\beta$ mass in Remark~\ref{rem:ADMalphabeta}.
\end{remark}
 
 In the case $\beta =0$, Proposition~\ref{prop:sqaf} implies that the metric is quasi-asymptotically flat in the sense of Definition~\ref{quasiflat}.
 
 \begin{corollary}\label{cor:decay}
 Suppose $g$ is a static spherically symmetric Lorentzian metric of the form \eqref{metric}, i.e., in local coordinates $(t,r,\theta,\phi)$ we can write
\[ 
 g = - e^{2\nu(r)} dt^2 + e^{2\lambda(r)} dr^2 + r^2 (d\theta^2 + \sin^2 \theta d\phi^2).
\] 
 If for some $\Lambda \geq 0$ the functions $\lambda$ and $\nu$ satisfy
\begin{equation}\label{ss1}
\begin{split}
 \nu'(r) &= o(r^{-\frac{1}{2}}), \\
 e^{2\lambda(r) - 2\Lambda} - 1 &= o(r^{-\frac{1}{2}}),
\end{split}
\end{equation}
as $r \to \infty$, then there exists a coordinate system $(\tau,\xi,\theta,\varphi)$ such that $\tilde g_{\mu\nu} = o (\xi^{-\frac{1}{2}})$ as $\xi \to \infty$. In particular, $g$ is quasi-asymptotically flat in the sense of Definition~\ref{quasiflat} with deficit angle \[(1-\alpha)\pi = e^{-2\Lambda}\pi.\]
\end{corollary}

\subsection{Beyond the ADM mass} \label{ssec:adma}

 Let us recall the definition of the ADM mass. For asymptotically flat metrics, the spatial part should satisfy
\[
 h_{ij} = \delta_{ij} + o(r^{-\frac{1}{2}}), \qquad \partial_k h_{ij} = O(r^{-\frac{3}{2}}), \qquad \text{as } r \to \infty.
\]
 The associated ADM mass can then be defined by the asymptotic behavior at spatial infinity,
\begin{align}\label{mADM}
 m_{\mathrm{ADM}}(h) 
&= \frac{1}{16 \pi} \lim_{R \to \infty} \int_{S_R(0)} (\partial^l h^i{}_{l} - \partial^i h^l{}_{l}) \, dS_i,
\end{align}
 where $S_R(0)$ is a $2$-sphere with radius $R$ and $dS_i$ are the Euclidean coordinate surface elements, i.e., $dS_i = \frac{x_i}{r} dx^1 \wedge \ldots \wedge \widehat{dx^i} \wedge \ldots \wedge dx^n$.
 The ADM mass exists and is finite if the scalar curvature $R(h)$ is integrable \cites{Bar:mass,Chr:mass}. Moreover, it is a geometric invariant (i.e., coordinate invariant) that is always nonnegative and zero only for the Minkowski metric \cites{SY1:pmt,SY2:pmt,W:pmt}.

 For general spherically symmetric Riemannian metrics of the form
\[ h= a(r) dr^2 + b(r) r^2 (d\theta^2 + \sin \theta^2 d\phi^2), \]
with $a,b$ differentiable and satisfying the above decay
$a-1 = o(r^{-\frac{1}{2}})$, $b-1 = o(r^{-\frac{1}{2}})$, $a_r = O(r^{-\frac{3}{2}})$ and $b_r = O(r^{-\frac{3}{2}})$, one
obtains (see, for example, \cite{Chr:energynotes}*{p.\ 12})
\begin{align}\label{ADMss}
 m_\mathrm{ADM}(h) = \frac{1}{2} \lim_{r\to\infty} \left( (a-b) r - b_r r^2 \right).
\end{align}
In particular, the ADM mass of the Schwarzschild metric coincides with the mass $m$ of the black hole.

\medskip
 Let us consider the problem of convergence for the integral in \eqref{mADM} in the case of quasi-asymptotically flat spacetimes. The asymptotic behavior of a spherically symmetric, quasi-asymptotically flat metric $g$ is dominated by the corresponding SAF$\alpha$ metric $g^\alpha$ \eqref{g0}. The spatial part $h^\alpha$ of $g^\alpha$, i.e.,
\[
 h^\alpha = d\xi^2 + (1-\alpha) \xi^2 (d\theta^2 + \sin^2 \theta d\phi^2),
\]
 has scalar curvature
\[
 R(h^\alpha) = \frac{2 \alpha}{1-\alpha} \xi^{-2} > 0, \qquad \xi \in (0,\infty). % not in L^1
\]
 Since $R(h^\alpha) \not \in L^1 (\RR^3 \setminus B_R(0))$ for any $R>0$, the ADM mass of $h^\alpha$ and thus $h$ is infinite \cite{Bar:mass}.
 
 One advantage of proving quasi-asymptotic flatness for a given spacetime is the availability of another concept of mass, the so-called ADM$\alpha$ mass for the spatial part %$h = h^0 + \tilde h$
 of the spacetime metric $g$. This natural generalization of the ADM mass for AF$\alpha$ slices can be defined using the background metric $h^\alpha$, again following the work of Nucamendi and Sudarsky \cite{NS:quasiasymp} with a slightly weaker falloff rate. They introduced the ADM$\alpha$ mass in the framework of Einstein--scalar theory.

 \begin{definition}[ADM$\alpha$ mass \cite{NS:quasiasymp}]\label{def:adma}
 Suppose $h^\alpha$ is the spatial SAF$\alpha$ metric defined for the hypersurface $\Sigma_\tau$ (i.e., such
 that $\tau = \mathrm{constant}$) of \eqref{g0} and
\[
 h = h^\alpha + \tilde h
\]
 is a spatial AF$\alpha$ metric with $\tilde h_{ij} = o(\xi^{-\frac{1}{2}})$ and $\partial_k \tilde h_{ij} = O(\xi^{-\frac{3}{2}})$ as $\xi \to\infty$.
 The ADM$\alpha$ mass of $h$ is defined by
\begin{align}\label{mADMa}
 m_{\mathrm{ADM}\alpha}(h) = \frac{1}{16 \pi (1-\alpha)} \lim_{\xi\to\infty} \int_{S_\xi(0)}
(h^{\alpha ik} h^{\alpha jl} - h^{\alpha ij} h^{\alpha kl}) \nabla^\alpha_j (h_{kl}) \, dS_i,
\end{align}
 where $\nabla^\alpha$ denotes the covariant derivative associated to $h^\alpha$ and $dS^i$ the $i$-th coordinate surface element with respect to $h^\alpha$.
\end{definition}

\begin{remark}[Relation to ADM mass]
 In the following sense, the ADM$\alpha$ mass is an extension of the ADM mass. If $\alpha = 0$, then $h^{0ij}=\delta^{ij}$ is just the reference Euclidean metric and
\begin{align*}
 (h^{0ik} h^{0jl} - h^{0ij} h^{0kl}) \nabla^0_j (h_{kl})
&= \delta^{ik}\delta^{jl} \partial_j h_{kl} - \delta^{ij}\delta^{kl} \partial_j h_{kl} = \partial^l h^i{}_{l} - \partial^i h^l{}_{l},
\end{align*}
thus \eqref{mADMa} yields exactly \eqref{mADM}.
\end{remark}

\begin{remark}
The ADM$\alpha$ mass coincides with the parameter $M$ of global monopole spacetimes
studied by Barriola and Vilenkin \cite{BV:monopole}.
\end{remark}

\begin{remark}[The ADM$\alpha$ mass is a geometric invariant of $(\Sigma,h)$.]\label{rem:ADMalpha}
 Nucamendi and Sudarsky proved in \cite{NS:quasiasymp}*{Section 4} that the ADM$\alpha$ is coordinate invariant given their slightly stricter setting with $\tilde h_{ij} = O(\xi^{-1})$ and $\partial_k \tilde h_{ij} = O(\xi^{-2})$. To see that the proof extends to our Definition~\ref{def:adma} it is crucial to note that $h$ in ``Cartesian'' coordinates $x^i$, $y^i$ (which are assumed to preserve the asymptotic behavior) reads
\begin{align*}
 h^{(x)}_{ij} &= (1-\alpha)\delta_{ij} + \alpha \frac{x_ix_j}{\xi^2} + A_{ij}, \\
 h^{(y)}_{ij} &= (1-\alpha)\delta_{ij} + \alpha \frac{y_iy_j}{\xi^2} + B_{ij},
\end{align*}
 now with
\begin{align}\label{AB}
 A_{ij} &= o(\xi^{-\frac{1}{2}}), & B_{ij} &= o(\xi^{-\frac{1}{2}}), \\
 \partial_k A_{ij} &= O(\xi^{-\frac{3}{2}}), & \partial_k B_{ij} &= O(\xi^{-\frac{3}{2}}),
\end{align}
as compared to $O(\xi^{-1})$ and $O(\xi^{-2})$ in \cite{NS:quasiasymp}*{Eq.\ (39)}. This follows directly from \cite{Chr:mass}*{Lemma 1} and the fact that the SAF$\alpha$ metric satisfies
\[
 h^\alpha_{ij} = (1-\alpha)\delta_{ij} + \alpha \frac{x_ix_j}{\xi^2}.
\]
 Lemma 1 in \cite{NS:quasiasymp}*{p.\ 1315} follows also from the decay assumptions~\eqref{AB}. In fact, $|A_{ab}| \leq \frac{C}{\xi^\gamma}$ for some $\gamma>0$ and sufficiently large $\xi$ yields the same result.
 Lemma 2 in \cite{NS:quasiasymp}*{p.\ 1316} follows too, in fact it can be improved to only require
$
 \eta^a = o(\xi^{\frac{1}{2}})$ and $\frac{\partial \eta^a(y)}{\partial y^b} = o(\xi^{-\frac{1}{2}})$.
 The final result for our weaker decay rates used in Definition~\ref{def:adma} follows from the Theorem in \cite{NS:quasiasymp}*{p.\ 1319 ff}, and it remains to be checked whether all the same terms can still be eliminated.
\end{remark}

\begin{remark}[Non-positivity of the ADM$\alpha$ mass.]
 Unlike the ADM mass for asymptotically flat spacetimes, the ADM$\alpha$ mass is not nonnegative. Indeed, it can be negative, depending on the choice of the reference metric (which corresponds to adding a constant). The more crucial point of whether the ADM$\alpha$ is generally bounded from below is still open. We refer to \cite{NS:quasiasymp,NS:blackholes} for a discussion of these issues.
\end{remark}

\medskip
 In what follows, we derive a simpler notion of the ADM$\alpha$ mass if the metric is spherically symmetric, in analogy to the expression \eqref{ADMss} in the asymptotically flat situation.

\begin{lemma}\label{lemma:ADMss}
 Suppose $h$ is a spherically symmetric, quasi-asymptotically flat, Riemannian $3$-metric of the form
\[
 h = a(\xi) \, d\xi^2 + b(\xi) \, (1-\alpha) \xi^2 (d\theta^2+\sin^2\theta d\phi^2), \qquad \alpha \in [0,1).
\]
%with sufficient decay in the components $a$ and $b$.
 Then the ADM$\alpha$ mass of $h$ is
\begin{align}\label{mADMss}
 m_{\mathrm{ADM}\alpha}(h) = \lim_{\xi \to\infty} \frac{1}{2} \left( \xi (a-b) - \xi^2 \partial_\xi b \right).
\end{align}
 Therefore, if $a - b = O(\xi^{-1})$ and $\partial_\xi b = O(\xi^{-2})$, then $m_{\mathrm{ADM}\alpha}(h)$ is finite.
\end{lemma}

\begin{proof}
 Recall that the spatial part of the SAF$\alpha$ metric is
\[
 h^\alpha = d\xi^2 + (1-\alpha)\xi^2 (d\theta^2+\sin^2\theta d\phi^2).
\]
 In Cartesian coordinates,
\[
 x_1 = \xi \sin \theta \cos \phi, \quad x_2 = \xi \sin \theta \sin \phi, \quad x_3 = \xi \cos \theta,
\]
 we have the metric components
\begin{align}\label{h0}
 h^\alpha_{ij} &= (1-\alpha) \delta_{ij} + \alpha \frac{x_ix_j}{\xi^2},
\end{align}
 since $d\xi^2 = \frac{x_ix_j}{\xi^2}dx^idx^j$. The components of the inverse metric $h^{-1}$ are
\begin{align}\label{inverseh0}
 h^{\alpha ij} = \frac{1}{1-\alpha} \left( \delta^{ij} - \alpha \frac{x^ix^j}{\xi^2} \right),
\end{align}
 thus, we can already compute the first term, $h^{\alpha ik}h^{\alpha jl}-h^{\alpha ij}h^{\alpha kl}$, in the integral of \eqref{mADMa}.

 It remains to compute the covariant derivative of $h$. Since $\nabla^\alpha$ is the Levi-Civita connection with respect to $h^\alpha$, by definition,
$\nabla^\alpha_k h^\alpha_{ij} = 0$ for all $i,j,k$. Therefore, we prefer to write $h$ as a perturbation of $h^\alpha$, that is,
\[
 h = b h^\alpha + (a-b) d\xi^2.
\]
 Hence,
\begin{align}\label{nablah}
 \nabla^\alpha h &= (\nabla^\alpha b) h^\alpha + \nabla^\alpha ((a-b) dr^2) \nonumber \\
&= (\nabla^\alpha b) (h^\alpha -d\xi^2) + (\nabla^\alpha a) d\xi^2 + (a-b) (\nabla^\alpha d\xi^2) \nonumber \\
&= (1-\alpha) (\nabla^\alpha b) (\delta - d\xi^2) + (\nabla^\alpha a) dr^2 + (a-b) (\nabla^\alpha dr^2) \nonumber \\
&= (1-\alpha) (\nabla^\alpha b) \delta + (\nabla^\alpha a - (1-\alpha) \nabla^\alpha b) d\xi^2 + (a-b) (\nabla^\alpha d\xi^2),
\end{align}
 where $\delta$ is the standard Euclidean metric $\delta = dx^2+dy^2+dz^2$. Since $b$ is just a function,
\[
 \nabla^\alpha_j b = \partial_j b = \frac{db}{d\xi} \frac{\partial \xi}{\partial x_j} = \partial_\xi b \frac{x_j}{\xi},
\]
 and similarly for $a$. For any $(0,2)$-tensor field, 
$
 \nabla^\alpha_j h_{kl} = \partial_j h_{kl} - {}^{\alpha} \G{m}{jk} h_{ml} - {}^{\alpha} \G{m}{jl} h_{mk},
$
 with Christoffel symbols
\begin{align}\label{eq:Ga}
 {}^{\alpha} \G{m}{jk} = \frac{\alpha}{\xi^2} \left( \delta_{jk} x^m - \frac{x^mx_jx_k}{r^2} \right).
\end{align}
 In particular, for $d\xi^2 = \frac{x_kx_l}{\xi^2} dx^kdx^l$,
\begin{align*}
 \nabla^\alpha_j (d\xi^2)_{kl}
&= \partial_j (d\xi^2)_{kl} - {}^{\alpha} \G{m}{jk} (d\xi^2)_{ml} - {}^{\alpha} \G{m}{jl} (d\xi^2)_{mk} \\
&= \delta_{kj} \frac{x_l}{\xi^2} + \delta_{lj} \frac{x_k}{\xi^2} - 2 \frac{x_j x_k x_l}{\xi^4}
 - \frac{\alpha}{\xi^2} \left( \delta_{jk} x_l  - \frac{x_jx_kx_l}{\xi^2} \right) - \frac{\alpha}{\xi^2} \left( \delta_{jl} x_k  - \frac{x_jx_kx_l}{\xi^2} \right) \\
&= \frac{1-\alpha}{\xi^2} \left( \delta_{jk} x_l + \delta_{jl} x_k - 2 \frac{x_jx_kx_l}{\xi^2} \right).
\end{align*}
 Thus, by \eqref{nablah},
\begin{align}\label{nablahh}
 \nabla^\alpha_j h_{kl}
&= (1-\alpha) \partial_\xi b \frac{x_j}{\xi} \delta_{kl} + (\partial_\xi a -(1-\alpha)\partial_\xi b) \frac{x_jx_kx_l}{\xi^3} \nonumber \\
&\quad + \frac{1-\alpha}{\xi^2}  (a-b) \left( \delta_{jk} x_l + \delta_{jl} x_k - 2 \frac{x_jx_kx_l}{\xi^2} \right).
\end{align}
 Together with \eqref{inverseh0}, we can now compute the $i$-th component of the integrand in \eqref{mADMa},
\[
 T^i := (h^{\alpha ik} h^{\alpha jl} - h^{\alpha ij} h^{\alpha kl}) \nabla^\alpha_j (h_{kl}).
\]
 For the first term in \eqref{nablahh}, we obtain
\begin{align*}
 T^i_{(1)} & = (h^{\alpha ik} h^{\alpha jl} - h^{\alpha ij} h^{\alpha kl}) (1-\alpha) b_r \frac{x_j}{\xi} \delta_{kl} \\
&= (1-\alpha) \frac{\partial_\xi b}{\xi} \left(h^{\alpha ik} x_k - h^{\alpha ij} x_j \frac{3-\alpha}{1-\alpha}\right)
= - 2 \partial_\xi b \frac{x^i}{\xi},
\end{align*}
 and for the second one
\begin{align*}
 T^i_{(2)} &= (\partial_\xi a -(1-\alpha) \partial_\xi b) (h^{\alpha ik} h^{\alpha jl} - h^{\alpha ij} h^{\alpha kl}) \frac{x_jx_kx_l}{\xi^3} = 0,
\end{align*}
 since $h^{\alpha ij}x_j = x^i$, $h^{\alpha kl}\frac{x_kx_l}{\xi^2} = 1$ and $h^{\alpha kl} \delta_{kl} = \frac{3-\alpha}{1-\alpha}$.
 Finally, the third term of the $i$-th component of the integrand is
\begin{align*}
 T^i_{(3)}
&= \frac{1-\alpha}{\xi^2} (a-b) (h^{\alpha ik} h^{\alpha jl} - h^{\alpha ij} h^{\alpha kl}) \left( \delta_{jk} x_l + \delta_{jl} x_k - 2 \frac{x_jx_kx_l}{\xi^2} \right) \\
&= \frac{1-\alpha}{\xi^2} (a-b) \left( 0 + x^i \frac{2}{1-\alpha} - 0 \right)
= 2 (a-b) \frac{x^i}{\xi^2}
\end{align*}
Since $\sqrt{\det{h^\alpha}} = (1-\alpha)$, the $i$-th coordinate surface element reads
\[
dS_i = \frac{x_i}{\xi} \sqrt{\det(h^\alpha)} dx^1 \wedge \ldots \wedge \widehat{dx^i} \wedge \ldots \wedge dx^n = (1-\alpha) \frac{x_i}{\xi} dx^1 \wedge \ldots \wedge \widehat{dx^i} \wedge \ldots \wedge dx^n, \] 
 where the notation $\widehat{dx^i}$ means that $dx^i$ is missing.
 Therefore, including the component $\frac{x_i}{\xi}$ of the $i$-th coordinate surface element, 
\begin{align}\label{T}
 T^i \frac{x_i}{\xi} = - 2 \partial_\xi b + 2 \frac{a-b}{\xi},
\end{align}
 and we obtain \eqref{mADMss} from \eqref{mADMa} and \eqref{T},
\begin{align*}
 m_{\mathrm{ADM}\alpha} &= \frac{1}{16 \pi (1-\alpha)} \lim_{\xi\to\infty} \int_{S_\xi(0)} \left( 2 \frac{a-b}{\xi} - 2 \partial_\xi b \right) (1-\alpha) dS \\
&= \frac{1}{8 \pi (1-\alpha)} \lim_{\xi\to\infty} \int_0^{\pi} \int_0^{2\pi} \left(\frac{a-b}{\xi} - \partial_\xi b \right) (1-\alpha) \xi^2 \sin\theta \, d\theta d\phi \\
&= \frac{1}{2} \lim_{\xi\to\infty} \left(\xi (a-b) - \xi^2 \partial_\xi b \right). \qedhere
\end{align*}
\end{proof}

 A more special case is the following, where the ADM$\alpha$ is already built in the construction of the metric, just like the ADM mass of the black hole is built in the standard expression of the Schwarzschild metric (cf.\ also \cite{NS:blackholes}*{Section2}).

\begin{corollary}\label{cor:MADM}
 Suppose $a(\xi)=\left( 1-\frac{2M(\xi)}{r}\right)^{-1}$ with $M(\xi) = o(\xi)$ as $r \to \infty$ and $b(\xi)=1$, then Lemma~\ref{lemma:ADMss} implies that the metric 
\[
 h= \left( 1-\frac{2M(\xi)}{\xi}\right)^{-1} d\xi^2 + (1-\alpha)\xi^2 (d\theta^2 + \sin^2\theta d\phi^2)
\]
 has ADM$\alpha$ mass
\[
 m_{\mathrm{ADM}\alpha}(h) =\lim_{\xi\to\infty} M(\xi).
\]
\end{corollary}

\begin{proof}
 By the assumption $M(\xi) = o(\xi)$, for each $n \in \mathbb{N}$ exists $\xi_n > 0$ such that
\[ 1 - \frac{1}{n} \leq 1 - \frac{2M(\xi)}{r} \leq 1 + \frac{1}{n} \]
for all $\xi \geq \xi_n$. Therefore, by Lemma~\ref{lemma:ADMss}, for all $n \in \mathbb{N}$,
 \begin{align*}
 m_{\mathrm{ADM}\alpha}(h) &= \frac{1}{2}  \lim_{\xi\to\infty} (a-b) \xi 
= \lim_{\xi \to \infty} \frac{M(\xi)}{1-\frac{2M(\xi)}{\xi}} 
\leq  \frac{n}{n-1} \lim_{\xi\to\infty} M(\xi),
 \end{align*}
 and similarly
\[
 m_{\mathrm{ADM}\alpha}(h) \geq \frac{n}{n+1} \lim_{\xi\to\infty} M(\xi),
\]
 which yields the desired result by the Squeeze Theorem.  
\end{proof}

\begin{remark}[ADM$\alpha\beta$ mass]\label{rem:ADMalphabeta}
 As in the quasi-asymptotically flat case, one should have a geometrically invariant ADM$\alpha\beta$ mass, that takes the slower convergence rate into account. One expects to obtain a mass that measures the deviation to
 \[
  h^{(\alpha,\beta)} = \xi^{2\beta} \left( d\xi^2 + (1-\alpha) \xi^2 (d\theta^2 + \sin^2 \theta d\phi^2) \right).
 \]
 An ad hoc candidate would be
  \begin{align}\label{madmab}
  m_{\text{ADM}\alpha\beta} (h) = \frac{1}{16 \pi (1-\alpha)} \lim_{\xi \to \infty} \xi^{-\beta} \int_{S_\xi(0)}
(h^{(\alpha,\beta) ik} h^{(\alpha,\beta) jl} - h^{(\alpha,\beta) ij} h^{(\alpha,\beta) kl}) \nabla^{(\alpha,\beta)}_j (h_{kl}) \, dS_i,
 \end{align}
 where $dS_i$ is the $i$-th coordinate surface element with respect to $h^{(\alpha,\beta)}$. 
 Clearly, if $\beta = 0$ \eqref{madmab} is just the ADM$\alpha$ mass of $h$, and if $\alpha=\beta=0$ it is the ADM mass. In fact, we expect an even more direct relation as appears in the spherically symmetric case in Remark~\ref{rem:ADMalphabetass}.
 \end{remark}  
 
 To obtain a rigorous geometrically invariant definition of such an ADM$\alpha\beta$ mass for scaled quasi-asymptoti\-cally flat metrics, one may follow the steps outlined by Michel~\cite{M:mass}. We will not pursue such a rigorous definition further in this paper, but rather provide arguments for why an investigation of such slowly converging spacetimes is useful in the first place, and how a notion of generalized ADM mass helps to control their asymptotic behavior.

\begin{remark}[ADM$\alpha\beta$ mass for the Minkowski metric] 
 In the Introduction, in Remark~\ref{rem:af}, we mentioned that the Minkowski metric can be rewritten as a SAF$\alpha\beta$ metric $g^{(1-1/(1+\beta)^2),\beta)}$ for any $\beta \geq 0$ if we choose a radial coordinate $\xi = \sqrt[1+\beta]{(1+\beta)r}$.
 For arbitrary $\alpha$ and $\beta$, we obtain that
 \[
  (h^{(\alpha,\beta) ik} h^{(\alpha,\beta) jl} - h^{(\alpha,\beta) ij} h^{(\alpha,\beta) kl}) \nabla^{(\alpha,\beta)}_j (h^0_{kl}) = 
  \frac{2((1-\alpha)(1+\beta)^2-1)}{\xi^{1+2\beta} (1-\alpha)(1+\beta)},
 \]
 for the integrand of the ADM$\alpha\beta$ mass of $h^0$, which in general may not be integrable. If, however, $1-\alpha = \frac{1}{(1+\beta)^2}$ then the integrand vanishes since $h^0$ is then the SAF$\alpha\beta$ metric $h^{(\alpha,\beta)}$ and
 \[
  m_{\mathrm{ADM}\alpha\beta} (h^0) = 0.
 \]
\end{remark}

\begin{remark}[ADM$\alpha\beta$ mass in spherical symmetry] \label{rem:ADMalphabetass}
 A comparison to $h^\alpha$ helps to simplify the formula \eqref{madmab}. Although $\xi$ is not the area radius (but a scaled version thereof, see Remark~\ref{rem:af}), we consider the same ``scaled'' Cartesian coordinates given by
\[
 y_1 = \xi \sin \theta \cos \phi, \qquad y_2 = \xi \sin \theta \sin \phi, \qquad y_3 = \xi \cos \theta.
\]
 The metric and inverse metric components are therefore
\[
 h^{(\alpha,\beta)}_{ij} = \xi^{2\beta} h^\alpha_{ij} = \xi^{2\beta} \left( (1-\alpha) \delta_{ij} + \alpha \frac{y_i y_j}{r^2} \right),
\]
and
\[
 h^{(\alpha,\beta)ij} = \xi^{-2\beta} h^{\alpha ij} = \frac{1}{\xi^{2\beta}(1-\alpha)} \left( \delta^{ij} - \alpha \frac{y^i y^j}{\xi^2} \right),
\]
 where we used the already derived expressions \eqref{h0} and \eqref{inverseh0} for the spatial components of the SAF$\alpha$ metric $g^\alpha$ read we can simplify the first factor in the integrand of \eqref{madmab} to
\begin{align}\label{hprod}
 h^{(\alpha,\beta) ik} & h^{(\alpha,\beta) jl} - h^{(\alpha,\beta) ij} h^{(\alpha,\beta) kl} = \xi^{-4\beta} (h^{\alpha ik} h^{\alpha jl} - h^{\alpha ij} h^{\alpha kl}) \\
 &= \frac{1}{\xi^{4\beta} (1-\alpha)^{2}} \left( \delta^{ik} \delta^{jl} - \delta^{ij} \delta^{kl} \right)
 - \frac{\alpha}{\xi^{2(1+2\beta)}(1-\alpha)^2} \left( \delta^{ik}y^jy^l + \delta^{jl}y^i y^k - \delta^{ij}y^ky^l - \delta^{kl}y^iy^j \right) \nonumber
\end{align} 
 Similarly, to simplify the term
\[
 \nabla^{(\alpha,\beta)}_j h_{kl} = \partial_j h_{kl} - {}^{(\alpha,\beta)} \G{m}{jk} h_{ml} - {}^{(\alpha,\beta)} \G{m}{jl} h_{mk}
\]
 we utilize the relation of the Christoffel symbols of $h^{(\alpha,\beta)}$ and $h^\alpha$ and the formula for ${}^{\alpha} \G{m}{jk}$ obtained in \eqref{eq:Ga}, i.e.,
 \begin{align*}
  {}^{(\alpha,\beta)} \G{m}{jk} &= (1+\beta) \, {}^{\alpha} \G{m}{jk} + \frac{\beta}{\xi^2} \Big( y_k \delta^m_j + y_j \delta^m_k - y^m \delta_{jk} \Big) \\
  &= (1+\beta) \frac{\alpha}{\xi^2} \Big( \delta_{jk} y^m - \frac{y^m y_jy_k}{\xi^2} \Big) + \frac{\beta}{\xi^2} \Big( y_k \delta^m_j + y_j \delta^m_k - y^m \delta_{jk} \Big),
 \end{align*}
 or also written as
 \begin{align*}
  {}^{(\alpha,\beta)} \G{m}{jk} = {}^{\alpha} \G{m}{jk} + \frac{\beta}{\xi^2} \Big( y_k \delta^m_j + y_j \delta^m_k - (1-\alpha) y^m \delta_{jk} - \alpha \frac{y^m y_j y_k}{\xi^2} \Big).
 \end{align*}

 Hence $\nabla^{(\alpha,\beta)}_j h_{kl}$ can be written as
\begin{align*}
  \nabla^{(\alpha,\beta)}_j h_{kl} = (1+\beta) \nabla^{\alpha}_j h_{kl} - \beta \Big( \partial_j h_{kl}
  - \frac{1}{\xi^2} (y_k h_{jl} +y_l h_{jk} + 2 y_j h_{kl} - \delta_{jk} y^m h_{ml} -\delta_{jl} y^m h_{mk}) \Big)
\end{align*}
or, alternatively, as
\begin{align*}
 \nabla^{(\alpha,\beta)}_j h_{kl} &= \nabla^{\alpha}_j h_{kl} - \frac{\beta}{\xi^2} \Bigg( \Big( y_k \delta^m_j + y_j \delta^m_k - (1-\alpha) y^m \delta_{jk} - \alpha \frac{y^m y_j y_k}{\xi^2} \Big) h_{ml} \\
  &\qquad\qquad\qquad\quad + \Big( y_l \delta^m_j + y_j \delta^m_l - (1-\alpha) y^m \delta_{jl} - \alpha \frac{y^m y_j y_l}{r^2} \Big) h_{mk} \Bigg) \\
  &= \nabla^{\alpha}_j h_{kl} - \frac{\beta}{r^2} \underbrace{\Bigg( 2 y_j h_{kl} + y_k h_{jl} + y_l h_{jk} - (1-\alpha) y^m ( \delta_{jk} h_{ml} + \delta_{jl} h_{mk}) - \alpha \frac{y^m y_j}{r^2} (y_k h_{ml} + y_l h_{mk})  \Bigg)}_{=: G^{(\alpha,\beta)}(h)_{jkl}}.
\end{align*}

The above manipulations hold for any Riemannian metric $h$.
For a spherically symmetric AF$\alpha\beta$ metric $\tilde h$ of the form
\begin{align}\label{tildeh}
  \tilde h &= \xi^{2\beta} \left( a(\xi) d\xi^2 + b(\xi) (1-\alpha) \xi^2 (d\theta^2 + \sin^2 \theta d\phi^2) \right) \nonumber \\
           &= b h^{(\alpha,\beta)} + (a-b) \xi^{2\beta} d\xi^2,  
 \end{align}
we can therefore directly use the computations of Lemma~\ref{lemma:ADMss} to simplify \eqref{madmab}, since $\tilde h = \xi^{2\beta} h$ with $h$ as in Lemma \ref{lemma:ADMss}.
It follows that
\begin{align}\label{term2}
  \nabla^{(\alpha,\beta)} \tilde h &= (1-\alpha) \xi^{2\beta} ( \nabla^{(\alpha,\beta)} b) \delta + \xi^{2\beta} (\nabla^{(\alpha,\beta)} a - (1-\alpha) \nabla^{(\alpha,\beta)} b) d\xi^2 \nonumber \\
  & \quad + (a-b) (\nabla^{(\alpha,\beta)} \xi^{2\beta} d\xi^2) \nonumber \\
  &= \xi^{2\beta} \nabla^{\alpha} h + (\nabla \xi^{2\beta}) (a-b) d\xi^2 + \xi^{2\beta} (a-b) (\nabla^{(\alpha,\beta)}d\xi^2-\nabla^{\alpha} d\xi^2),
 \end{align}
 where $\nabla_j \xi^{2\beta} = \partial_j \xi^{2\beta} = 2 \beta \xi^{2\beta} \frac{x_j}{\xi^2}$ and $\nabla^{(\alpha,\beta)}d\xi^2-\nabla^{\alpha} d\xi^2 = - \frac{\beta}{\xi^2} G^{(\alpha,\beta)}(d\xi^2)_{jkl}$.
We consider the integrand
 \begin{align*}
  \widetilde T^i &= (h^{(\alpha,\beta) ik} h^{(\alpha,\beta) jl} - h^{(\alpha,\beta) ij} h^{(\alpha,\beta) kl}) \nabla_j^{(\alpha,\beta)} \tilde h_{kl} \\
      &= \xi^{-2\beta} (h^{\alpha ik} h^{\alpha jl} - h^{\alpha ij} h^{\alpha kl}) \left( \nabla_j^{\alpha} h_{kl} + 2\beta (a-b) \frac{y_jy_ky_l}{\xi^4} - (a-b) \frac{\beta}{\xi^2} G^{(\alpha,\beta)}(d\xi^2)_{jkl} \right)
 \end{align*}
By \eqref{nablahh} in the proof of Lemma~\ref{lemma:ADMss}, the first term in this product is known to be
 \begin{align*}
  \widetilde T^i_{(1)} = \xi^{-2\beta} T^i = \xi^{-2\beta} \left( - 2 \partial_\xi b \frac{y^i}{\xi} + 2 (a-b) \frac{y^i}{\xi^2} \right),
 \end{align*}
and the second term vanishes similar to $T^i_{(2)}$. To compute the last term, we note that $(d\xi^2)_{kl} = \frac{y_ky_l}{\xi^2}$ and
\[
 G^{(\alpha,\beta)}(d\xi^2)_{jkl} = - 2 \frac{y_j y_k y_l}{\xi^2} + (1-\alpha) ( \delta_{jk} y_l + \delta_{jl} y_k - 2 \frac{y_j y_k y_l}{\xi^2}), 
\]
of which the first term disappears in the product $\widetilde T_i$ (similar to $T^i_{(2)}$) and the second one simplifies as in $T^i_{(3)}$. Thus, the last term in $\widetilde T_i$ reads
\[
 \widetilde T^i_{(3)} = \xi^{-2\beta} (a-b) \frac{\beta}{\xi^2} (h^{\alpha ik} h^{\alpha jl} - h^{\alpha ij} h^{\alpha kl}) (1-\alpha) \delta_{jl} x_k = \xi^{-2\beta} 2 \beta (a-b) \frac{y^i}{\xi^2}.
\]
We combine all the terms and thus arrive at
\[
 \widetilde T^i = \xi^{-2\beta} \left( - 2 (\partial_\xi b) \frac{y^i}{\xi}  + 2 (1+\beta) (a-b) \frac{y^i}{\xi^2} \right),
\]
and thus,
\begin{align}\label{Ty}
 \widetilde T^i \frac{y_i}{\xi} = r^{-2\beta} \left( 2(1+\beta) \frac{a-b}{\xi} - 2 \partial_\xi b \right).
\end{align}

Finally, note that the components of the $i$-th coordinate surface element $dS_i$ with respect to $h^{(\alpha,\beta)}$ are
\[
 dS_i = \frac{y_i}{\xi} \sqrt{\det(h^{(\alpha,\beta)})} dy_1 \wedge \ldots \wedge \widehat{dy^i} \wedge \ldots \wedge dy^d = (1-\alpha) \xi^{d\beta} \frac{y_i}{\xi} dy_1 \wedge \ldots \wedge \widehat{dy^i} \wedge \ldots \wedge dy^d,
\]
where $d=3$ is the dimension.
The normal vectors coincide, and hence,
\begin{align}\label{dS}
 dS_i = (1-\alpha) \xi^{3\beta} d\widetilde S_i,
\end{align}
where $d\widetilde S_i$ is the $i$-th coordinate surface element of the Euclidean space (with $r$ replaced by $\xi$, and $x_i$ replaced by $y_i$).

Thus, combining \eqref{Ty} with \eqref{dS} yields
\begin{align}\label{mab}
    m_{\mathrm{ADM}\alpha\beta} (\tilde h) &= \frac{1}{16\pi(1-\alpha)} \lim_{\xi \to\infty} \xi^{-3\beta} \int_{S_\xi(0)}  \left( 2(1+\beta) \frac{a-b}{\xi} - 2 \partial_\xi b \right) d S_i \nonumber \\
    &= \frac{1}{8\pi(1-\alpha)} \lim_{\xi \to\infty} \xi^{-3\beta} \int_0^\pi \int_0^{2\pi}  \left( 2(1+\beta) \frac{a-b}{\xi} - 2 \partial_\xi b \right) (1-\alpha) \xi^{3\beta} \xi^2 \sin \theta d\theta d\phi \nonumber \\
    &= \frac{1}{2} \lim_{\xi \to \infty} \left( \xi (1+\beta) (a-b) - \xi^2 \partial_\xi b \right),
\end{align}
a formula for the ADM$\alpha\beta$ for spherically symmetric AF$\alpha\beta$ metrics of the form \eqref{tildeh}.
\end{remark}

%%%%%%%%%%%%%%%%%%%%%%%%%%%%%%%%%%%%%%%%%%%%%%%%%%%%%%%%%%%%%%%%%%%%%%%%%%%%%%%%%%%%%%%%
%%%%%% SECTION 3 %%%%%%%%%%%%%%%%%%%%%%%%%%%%%%%%%%%%%%%%%%%%%%%%%%%%%%%%%%%%%%%%%%%%%%%
%%%%%%%%%%%%%%%%%%%%%%%%%%%%%%%%%%%%%%%%%%%%%%%%%%%%%%%%%%%%%%%%%%%%%%%%%%%%%%%%%%%%%%%%

\section{Perfect fluids with linear equation of state}
\label{sec3}

 The linear relation $p = K \rho$, $K \in [0,1]$, between the pressure $p$ and the mass--energy density $\rho$ immediately implies that the static Einstein--Euler equations \eqref{staticEE1}--\eqref{staticEE2} in spherical symmetry can be reformulated as system of ordinary differential equations in $m$ and $\rho$,
\begin{subequations} \label{EE:linear}
\begin{align}
 m_r &= 4 \pi r^2 \rho, \label{EE1:linear} \\
 \rho_r &= - \frac{(1+K) \rho}{r - 2m} \left( 4 \pi r^2 \rho + \frac{m}{rK} \right). \label{EE2:linear}
\end{align}
\end{subequations}
 Even in this simplest setting, only very exceptional exact solutions are known~\cite{Iva:integrable}. More can be said about the asymptotic behavior of solutions. A geometric understanding of the asymptotic behavior of the resulting spherically symmetric, static spacetime metric
\begin{align}\label{metric1}
 g = - e^{2\nu(r)} dt^2 + e^{2\lambda(r)} dr^2 + r^2 (d\theta^2 + \sin^2 \theta d\phi^2),
\end{align}
 is the main goal of this Section. The metric components $\lambda,\nu$ can be derived by integrating \eqref{inteq1}--\eqref{inteq2} and yield (for $\nu(0) := 0$)
\begin{align}\label{munu}
 e^{2\nu(r)} = \left( \frac{\rho_0}{\rho(r)} \right)^{\frac{2K}{1+K}}, \qquad e^{2\lambda(r)} = \left( 1- \frac{2m}{r} \right)^{-1}.
\end{align}

\subsection{The initial value problem} \label{ss:ivp}

 Suppose we prescribe an central density $\rho_0>0$ and $K \in (0,1]$. According to Rendall and Schmidt \cite{RS:static}*{Theorem 2} and LeFloch and the second author \cite{BLF:trapped}*{Theorem 4.3} there exists a unique, smooth and positive global solution $(m,\rho)$ of \eqref{EE:linear} such that
\[
  \lim_{r\to0} m(r) = 0, \qquad \lim_{r\to0} \rho(r) = \rho_0.
\]
 The solution must have infinite extent since condition \eqref{inf} is violated, i.e., for any $p_0>0$,
\[
 \int_0^{p_0} \frac{dp}{\rho(p)+p} = \int_0^{p_0} \frac{dp}{(1+\frac{1}{K})p} = \infty. 
\]
 Early observations by Chandrasekhar~\cite{Chand} and others~\cites{Cha:linear,Mak:static} reveal that the mass function $m(r) = \int_0^r s^2 \rho(s) \, ds$ grows with $r^3$ near the center and linearly in $r$ near infinity. The asymptotic behavior for $r \to 0$ is
\begin{align*}
     m(r) &= \frac{4\pi}{3} \rho_0 r^3 - 4 \pi^2 \frac{(1+K)(1+3K)}{9K} \rho_0^2 r^5 + O(r^7), \\
     \rho(r) &= \rho_0 - 2 \pi \frac{(1+K)(1+3K)}{3K}\rho_0^2 r^2 + 8 \pi^2 \frac{(1+K)(1+3K)}{9K^2} \rho_0^3 r^4 - O(r^6),
\end{align*}
 according to the Taylor series expansion derived by differentiating \eqref{EE:linear} at the origin \cite{BLF:trapped}. It helps to observe that $\rho$ is an even function and $m$ is an odd function if we would consider solutions on the whole real line. See also \cite{HM} for higher-order terms of the mass function and linear barotropic and polytropic equations of state.

\subsection{Asymptotic behavior at infinity}

 In contrast to the initial behavior, less is known about the behavior of $m$ and $\rho$ as $r$ tends to infinity. We already know that $m$ is strictly increasing as $r \to \infty$, $\rho$ is strictly decreasing with $\lim_{r\to\infty} \rho(r) = 0$ and $r^2 \rho(r)$ remains bounded \cite{BLF:trapped}*{Theorem 4.3}. However, the solution is not asymptotically flat due to formula \eqref{ADMss} for the ADM mass for spherically symmetric metrics, which yields that
\[
 m_\textrm{ADM} = \lim_{r\to\infty} m(r) = \infty.
\]
 In order to still be able to say something about the behavior of the solution at radial infinity, we therefore need to have a better understanding of the growth rate of $m$ for large $r$.

\subsubsection{The singular solution}\label{ss:singsol}

 Na\"{\i}vely, in order to derive some asymptotics as $r \to \infty$, we make the Ansatz
\[
 m(r) = c_1 r^\alpha, \qquad \rho(r)= c_2 r^{\beta},
\]
 for some $\alpha,\beta,c_1,c_2 \in \mathbb{R}$. Plugged into the system \eqref{EE:linear}, this yields the exact solution
\begin{align}\label{singsol}
 m_\infty(r) := \frac{2K}{(1+K)^2 + 4K} r, \qquad \rho_\infty(r) = \frac{K}{2\pi ((1+K)^2 + 4K)} r^{-2}.
\end{align}
 Obviously, this solution is special and somewhat unphysical since the density blows up at the center. Because the trajectories of solutions cannot intersect, this singular solution is an upper bound for all regular solutions of \eqref{EE:linear} with central density $\rho_0 \in (0,\infty)$. %, see Figure~\ref{fig1}.

 Setting $\nu(1):=0$, integrating \eqref{inteq1} yields \[e^{2\nu(r)} = r^{\frac{4K}{1+K}},\] and \eqref{inteq2} yields \[e^{2\lambda(r)} = \left( 1 - \frac{2m_\infty(r)}{r} \right)^{-1} = \frac{(1+K)^2+4K}{(1+K)^2}.\] We write $g_\infty = g_\infty(K)$ for the singular metric \eqref{metric1} with sound speed $\sqrt{K} \in (0,1]$ corresponding to $(m_\infty,\rho_\infty)$, i.e.,
\begin{align*}
 g_\infty %&= -r^{\frac{4K}{1+K}} dt^2 + \left( 1 - \tfrac{2m_\infty}{r} \right)^{-1} dr^2 + r^2(d\theta^2 + \sin^2 \theta d\phi^2) \\
          &= -r^{\frac{4K}{1+K}} dt^2 +  \frac{(1+K)^2+4K}{(1+K)^2} dr^2 + r^2(d\theta^2 + \sin^2 \theta d\phi^2).
\end{align*}
 Although unphysical, these singular solutions play an important role from the geometric point of view.

\begin{proposition}\label{prop:quasi1}
 Fix $K \in (0,1]$. The singular solution \eqref{singsol} of the Einstein--Euler system~\eqref{EE:linear} with linear equation of state $p = K \rho$ is quasi-asymptotically flat. More precisely, up to a coordinate transformation it is the AF$\alpha$ metric
\begin{equation}\label{ginf:as}
\begin{split}
 g_\infty = \; & -d\tau^2 + d\xi^2 + (1-\alpha)\xi^2 (d\theta^2 + \sin^2 \theta d\phi^2) + \frac{4K}{1+K} \frac{\tau}{\xi} d\tau d\xi - \left(\frac{2K}{1+K} \frac{\tau}{\xi} \right)^2 d\xi^2,
\end{split}
\end{equation}
 with deficit angle
\[
 \alpha = \frac{4K}{(1+K)^2+4K}.
\]
The ADM$\alpha$ mass of each spatial slice $\Sigma_\tau$ vanishes.
\end{proposition}

Note that the deficit angle $(1-\alpha)\pi$ remains within the interval $\left[ \frac{\pi}{2}, \pi \right]$ for all $K$, and
hence is bounded away from $0$ uniformly for \emph{all} linear equations of state.

\begin{proof}
 Since $\nu'(r) = O(r^{-1})$ and $\lambda(r) = \Lambda$, it follows immediately from Corollary~\ref{cor:decay}, and the coordinate transformations $\tau = r^{\frac{2K}{1+K}}t$ and $\xi = e^\Lambda r$ used in the proof, that $g_\infty$ is quasi-asymptotically flat of the form \eqref{ginf:as}. The deficit angle is given by
 \[
  1-\alpha = e^{-2\Lambda} = \lim_{r\to\infty} e^{-2\lambda(r)} = \lim_{r\to\infty} \left( 1 - \tfrac{2m(r)}{r} \right) = 1 - \frac{4K}{(1+K)^2 + 4K} = \frac{(1+K)^2}{(1+K)^2 + 4K}.
 \]
 By Lemma~\ref{cor:MADM}, since $a(\xi) = 1 - \left( \frac{2K}{1+K} \frac{\tau}{\xi} \right)$ and $b=1$, the ADM$\alpha$ mass of
 \[
 h_\infty(\tau) =  g_\infty |_{\Sigma_\tau} = a(\xi) d\xi^2 + (1-\alpha) \xi^2 (d\theta^2 + \sin^2 \theta d\phi^2)
 \]
 at each spatial slice $\Sigma_\tau$ is
 \[
  m_{\mathrm{ADM}\alpha} (h_\infty(\tau)) = \lim_{\xi\to\infty} \frac{a(\xi) -b(\xi)}{2\xi} = -\lim_{\xi\to\infty} \frac{4K^2}{(1+K)^2} \frac{\tau^2}{\xi} = 0. \qedhere
 \]
\end{proof}

\begin{remark}
 Dadhich~\cites{D:iso1,D:iso2} recovered the above family of singular solutions as those spherically symmetric isothermal perfect fluids without boundary that are conformal to a Kerr--Schild metric. In this case, the latter is given by components
 \[
  g_{\mu\nu} = \eta_{\mu\nu} + 2 H l_{\mu} l_{\nu},
 \]
 where $\eta$ denotes the flat Minkowski metric, $H$ a constant nonzero scalar field and $l^{\mu}$ represents a null vector relative to $g$ and $\eta$.
 He called this geometric behavior ``minimally curved''.
\end{remark}

\subsubsection{Reformulation as a dynamical system}

 In 1972 Chandrasekhar~\cite{Chand} studied the asymptotic behavior of the system~\eqref{EE:linear} by reformulating the system using Milne variables. In the late 1990s, Makino reformulated \eqref{EE:linear} as an autonomous system and used plane dynamical systems theory, more precisely the Poincar\'e--Bendixson Theorem, to obtain that for $K = \frac{1}{3}$ the singular solution is the only element in the $\omega$-limit set and hence all regular solutions converge to it \cite{Mak:static,Mak:spiral}. While the case of the linear equation of state is not directly included in the dynamical systems analysis of Heinzle, R\"ohr and Uggla~\cite{HRU}, it can be seen as the limiting case $n\to\infty$ of relativistic polytropes \eqref{polytropicEOS}. The convergence to the only $\omega$-limit and fixed point, i.e., the singular solution, thus would also follow from their approach. As already mentioned in Section~\ref{ssec:dynamical} of the Introduction, however, the existing implicit reformulations as dynamical system cannot be applied directly, because they do not allow for a translation of a convergence rate in the original radial variable. Instead, while otherwise using a similar approach as in \cite{Mak:spiral}*{Section 2}, we utilize an explicit reformulation.

\begin{lemma}\label{lem:ab}
 Fix $K \in (0,1]$. The spherically symmetric Einstein--Euler system \eqref{EE:linear} is equivalent to the autonomous system
 \begin{subequations} \label{EEdyn}
  \begin{align}
  \dot a &= 1 - e^a + 2 e^b, \\
  \dot b &= \frac{1+7K}{2K} - \frac{1+3K}{2K} e^a + (1-K) e^b,
  \end{align}
 \end{subequations}
 where $t(r) = \log r$ and
 \begin{align} \label{def:ab}
  a(t) = - \log \left( 1 - \frac{2m(r(t))}{r(t)} \right), \qquad b(t) = \log \left( 4 \pi r(t)^2 \rho(r(t)) \right) + a(t). \qedhere
 \end{align}
\end{lemma}

\subsubsection{Asymptotic stability and convergence to the singular solution}

 The singular solution \eqref{singsol} transforms in the formulation of Lemma~\ref{lem:ab} to the constant singular solution $(a_\infty,b_\infty)$ with
\begin{align}\label{singsolab}
 a_\infty (t) = \log \frac{(1+K)^2+4K}{(1+K)^2}, \qquad b_\infty(t) = \log \frac{2K}{(1+K)^2}.  
\end{align}
 It plays the special role of the single $\omega$-limit point of the plane dynamical system \eqref{EEdyn}. In fact, it is a hyperbolic fixed point and we can analyze the stability of the nonlinear flow by linearizing the system around $(a_\infty,b_\infty)$.

\begin{lemma}%[Linearization around the singular solution]
\label{lem:sink}
 Fix $K \in (0,1]$. The singular solution \eqref{singsolab} is a fixed point and the only $\omega$-limit point of the plane dynamical system \eqref{EEdyn}, i.e., all solutions $(a,b)$ converge to $(a_\infty,b_\infty)$ as $t \to\infty$. It is a nonlinear hyperbolic sink for \eqref{EEdyn}, and the eigenvalues $\lambda_\pm$
 of the linearized equation are
 \[
  \lambda_\pm = - \frac{1+3K}{2(1+K)}  \pm i \frac{\sqrt{7+42K-K^2}}{2(1+K)},
 \]
 with $-1 \leq \Re \lambda_\pm < - \frac{1}{2}$.
\end{lemma}

\begin{proof}
 The fact that $(a_\infty,b_\infty)$ is the single $\omega$-limit point follows from the Poincar\'e--Bendixson Theorem by excluding the possibilities of orbits and other fixed points
 as in \cite{Mak:spiral}*{Section 2}.

 To compute the eigenvalues, let us write $x = (a,b)$ and $\dot x = F(x)$ for the dynamical system \eqref{EEdyn}. The linearization around $x_\infty = (a_\infty,b_\infty)$ is
 of the form
 \[
  \frac{dx(t)}{dt} = A_\infty x(t) + c_\infty,
 \]
 where
\begin{align}\label{DF}
 A_\infty := DF |_{(a_\infty,b_\infty)}
                                       = \begin{pmatrix}
                                       -e^{a_\infty} & 2 e^{b_\infty} \\
                                       -\frac{1+3K}{2K} e^{a_\infty} & (1-K) e^{b_\infty}
                                      \end{pmatrix}
                                        = \begin{pmatrix}
                                           - \frac{(1+K)^2 + 4K}{(1+K)^2} & \frac{4K}{(1+K)^2} \\
                                           - \frac{1+3K}{2K} \frac{(1+K)^2 + 4K}{(1+K)^2} & \frac{2K(1-K)}{(1+K)^2}
                                          \end{pmatrix},
\end{align}
 and 
\begin{align*}
 c_\infty := - A_\infty x_\infty = \begin{pmatrix}
                                    \log \frac{(1+K)^2 (2K)^\frac{4K}{(1+K)^2}}{((1+K)^2+4K)^\frac{(1+K)^2+4K}{(1+K)^2}} \\
                                    \log \frac{(1+K)^\frac{1+7K}{K} (2K)^\frac{2K(1-K)}{(1+K)^2}}{((1+K)^2+4K)^{\frac{1+3K}{2K}\frac{(1+K)^2+4K}{(1+K)^2}}}
                                   \end{pmatrix}.
\end{align*}
 The eigenvalues of $A_\infty$ are
\[
 \lambda_\pm = %-\alpha \pm i \beta :=
 - \frac{1+3K}{2(1+K)} \pm i \frac{\sqrt{7+42K-K^2}}{2(1+K)},
\]
 the corresponding eigenvectors are
\[
  u_\pm
     = \begin{pmatrix}
        K(1+8K-K^2) \mp i K(1+K)\sqrt{7+42K-K^2}\\
        (1+3K)((1+K)^2+4K)
       \end{pmatrix}.
\]
 Because the eigenvalues have negative real part $- \frac{1+3K}{2(1+K)} < - \frac{1}{2}$ (since $K>0$ the inequality is also strict), the singular solution $(a_\infty,b_\infty)$ is a hyperbolic sink.
\end{proof}

 We denote by $\varphi^t$ the nonlinear flow of $\dot x = F(x)$, i.e., $\varphi^t(x_0)$ is the solution $x(t)$ of $\dot x = F(x)$ with initial
 condition $x(0)=x_0$. Standard dynamical systems theory provides a control of the asymptotic behavior in the vicinity of the singular solution \eqref{singsolab}.

\begin{theorem}[Asymptotic stability in terms of $(a,b)$]\label{thm:asymp:ab}
 Fix $K \in (0,1]$. For every norm $|.|$ on $\mathbb{R}^2$ there exists a constant $C\geq 1$ and a neighborhood $U$ of the singular solution $x_\infty = (a_\infty,b_\infty)$ such that for any initial condition $x \in U$, the solution is defined for all $s \geq 0$ and for any $\varepsilon >0$,
\begin{align}\label{flow}
 | \varphi^s(x) - x_\infty | \leq C e^{-\left(\frac{1+3K}{2(1+K)}-\varepsilon\right) s} | x-x_\infty|, \qquad s \geq 0.
\end{align}
 Thus, in particular, the singular solution is asymptotically stable. 

 Moreover, there is a neighborhood $U$ around $x_\infty = (a_\infty,b_\infty)$ such that the flow $\varphi^s$ of $F$ is $C^1$-conjugate to the affine flow $s \mapsto x_\infty + e^{A_\infty s} (x-x_\infty)$, where $A_\infty = DF |_{x_\infty}$ \eqref{DF}. That is, there exists a $C^1$-diffeomorphism $h \colon U \to U$ such for $x,x_\infty + e^{A_\infty s} (x-x_\infty) \in U$ we have
\[
 \varphi^s(h(x)) = h(x_\infty + e^{A_\infty s} (x-x_\infty)).
\]
\end{theorem}

 Note that, instead of employing $\varepsilon >0$ in \eqref{flow}, we can write more rigidly
\[
 | \varphi^s(x) - x_\infty | \leq C e^{-\frac{s}{2}} | x-x_\infty|, \qquad s \geq 0,
\]
 to obtain an estimate independent of $K \in (0,1]$.

\begin{proof}
 Since $\Re \lambda_\pm = - \frac{1+3K}{2(1+K)} < - \frac{1}{2} < 0$ by Lemma~\ref{lem:sink}, the first part of the statement is due to the exponential contraction of the linear flow and the Gronwall inequality (see, for example, \cite{R}*{Theorem 5.1}).

 The conjugacy statement follows from the Hartman--Grobman theorem, or can also be proven directly since $(a_\infty,b_\infty)$ is a hyperbolic sink (see, for example, \cite{R}*{Theorem 5.2}). The fact that we obtain a $C^1$-diffeomorphism and not merely a homeomorphism $h$ follows from the smoothness of $F$ \cites{H:diffeo,S:diffeo}.
\end{proof}

\begin{corollary}[Asymptotic stability in terms of $(m,\rho)$]\label{cor:asymp}
 Fix $K \in (0,1]$. The asymptotic behavior of solutions $(m,\rho)$ to the system~\eqref{EE:linear} with initial data $\rho_0>0$ for $r \to \infty$ is
\begin{align*}
 m(r) &= \frac{2K}{(1+K)^2+4K} r + O\left(r^{\frac{1-K}{2(1+K)}+\varepsilon}\right), \\
 \rho(r) &= \frac{K}{2\pi ((1+K)^2+4K)} r^{-2} + O\left( r^{-\frac{5+7K}{2(1+K)} + \varepsilon} \right).
\end{align*}
\end{corollary}

\begin{proof}
\textbf{Step 1. Estimate $a(t) -a_\infty$.}
 Let us fix an initial density $\rho_0>0$. For every such $\rho_0$ there exists a unique global smooth solution $(m,\rho)$ to the system~\eqref{EE:linear} with infinite extent (see discussion in Section~\ref{ss:ivp}). In particular, we obtain the corresponding initial value $x_0 = (a_0,b_0)$ through
 \[
  a_0 := - \log \left( 1 - 2 m(1) \right), \qquad b_0 := \log (4 \pi \rho(1) ) + a_0,
 \]
 for the reformulated system \eqref{EEdyn} and a corresponding solution $(a,b)$ with $a(0)=a_0$ and $b(0)=b_0$. By Lemma~\ref{lem:ab},
 \[
  x(t) = (a(t),b(t)) \to x_\infty = (a_\infty,b_\infty) \qquad \text{as} ~ t \to\infty.
 \]
 Thus, there exists a $t_0 = t_0(K,\rho_0) >0$ such that for all $t\geq t_0$ the remaining solution $x(t)$ is in the neighborhood $U$ of $x_\infty$ obtained in Theorem~\ref{thm:asymp:ab}. By \eqref{flow}, since the flow satisfies $\varphi^{t} = \varphi^{t_0 + s} = \varphi^s \circ \varphi^{t_0}$ for $s = t -t_0$,
\begin{align*}
 | \varphi^t(x_0) - x_\infty | = | \varphi^s (x(t_0)) - x_\infty | \leq C e^{-\left(\frac{1+3K}{2(1+K)}-\varepsilon\right) s} | x(t_0) -x_\infty|.
\end{align*}
 If we replace the constant $C = C(K)$ by a constant $\widetilde C = \widetilde C(K,\rho_0)$ and assume without loss of generality that all elements $y$ in $U$ satisfy $|x_\infty-y| \leq 1$, then
\begin{align*}
 | \varphi^t(x_0) - x_\infty | \leq C  e^{-\left(\frac{1+3K}{2(1+K)}-\varepsilon\right) t}  e^{\left(\frac{1+3K}{2(1+K)}-\varepsilon\right) t_0}
 \leq \widetilde C e^{-\left(\frac{1+3K}{2(1+K)}-\varepsilon\right) t}
\end{align*}
 In particular, if we think of $|.|$ as the maximum norm in $\mathbb{R}^2$, then for all $\varepsilon > 0$,
 \begin{subequations}
\begin{align}
  | a(t) - a_\infty | &\leq \widetilde C e^{-\left(\frac{1+3K}{2(1+K)}-\varepsilon\right) t}, \label{a1} \\
  | b(t) - b_\infty | &\leq \widetilde C e^{-\left(\frac{1+3K}{2(1+K)}-\varepsilon\right) t}. \label{b1}
\end{align}
 \end{subequations}
\textbf{Step 2. Estimate $m(r) - m_\infty$.} By Definition~\eqref{def:ab} of $a$, 
\begin{align*}
 m(r) &= \frac{r}{2} \left( 1 - e^{-a(t(r))} \right) = \frac{r}{2} \left( 1 - \frac{(1+K)^2}{(1+K)^2+4K} e^{a_\infty -a(t(r))} \right) \\
 &= \frac{r}{2} \left( 1 - \frac{(1+K)^2}{(1+K)^2+4K} - \frac{(1+K)^2}{(1+K)^2+4K} ( e^{a_\infty -a(t(r))} - 1) \right) \\
 &= \frac{2K}{(1+K)^2+4K} r - \frac{(1+K)^2}{2((1+K)^2+4K)} r (e^{a_\infty -a(t(r))} - 1) \\
 &= m_\infty (r) + O\left(r^{1-\frac{1+3K}{2(1+K)}+\varepsilon}\right),
\end{align*}
 since for $r$ (and hence $t(r)$) sufficiently large
\begin{align*}
 \left| e^{a_\infty -a(t(r))} - 1 \right| &= \left| \sum_{n=1}^\infty \frac{(a_\infty - a(t(r)))^n}{n!} \right| \nonumber \\
 &= |a_\infty - a(t(r))| \sum_{n=0}^{\infty} \frac{|a_\infty - a(t(r))|^n}{(n+1)!}
 \leq |a_\infty - a(t(r))| e,
\end{align*}
 and thus, by \eqref{a1} of Step 1 there is a constant $C>0$ such that
\begin{align}\label{step1}
 \left| e^{a_\infty -a(t(r))} - 1 \right| \leq C e^{-\left(\frac{1+3K}{2(1+K)}-\varepsilon\right) t(r)} = C r^{-\left(\frac{1+3K}{2(1+K)}-\varepsilon\right)}.
\end{align}
 \textbf{Step 3. Estimate $\rho(r) - \rho_\infty$.} This follows from \eqref{b1} in Step 1, Step 2 and the Definition~\eqref{def:ab} of $b$. We have that
\begin{align*}
 \rho(r) &= \frac{1}{4\pi r^2} e^{b(t(r))-a(t(r))} = \frac{1}{4\pi r^2} e^{b(t(r))-b_\infty} e^{a_\infty - a(t(r))} e^{b_\infty-a_\infty} \\
        &= \frac{K}{2 \pi ((1+K)^2+4K)} r^{-2} e^{b(t(r))-b_\infty} e^{a_\infty - a(t(r))} \\
        &= \frac{K}{2 \pi ((1+K)^2+4K)} r^{-2} \left( (e^{b(t(r))-b_\infty} - 1) + 1\right) \left( (e^{a_\infty - a(t(r))} - 1) + 1\right) \\
        &= \rho_\infty(r) + O\left( r^{-2 - \frac{1+3K}{2(1+K)} +\varepsilon} \right),
\end{align*}
 by \eqref{step1} and the same estimate for the $b$-term.
\end{proof}

\subsection{Quasi-asymptotic flatness and ADM$\alpha$ mass}

 In Proposition~\ref{prop:quasi1} we have observed that the singular solution $(m_\infty,\rho_\infty)$ to the system~\eqref{EE:linear} is the SAF$\alpha$ metric with $\alpha = \frac{4K}{(1+K)^2+4K}$. The ADM$\alpha$ mass of this singular solution vanishes. Since by Corollary~\ref{cor:asymp} every solution $(m,\rho)$ to the initial value problem~\eqref{EE:linear} with $\rho_0>0$ converges to $(m_\infty,\rho_\infty)$, it is natural to expect that these solutions are also quasi-asymptotically flat. However, due to the slow convergence rate obtained in Corollary~\ref{cor:asymp} we cannot say whether the ADM$\alpha$ mass is even finite.

\medskip
 We first derive conditions on the mass function $m$ and density $\rho$ that imply that we are dealing with an quasi-asymptotically flat spacetime in the sense of Definition~\ref{quasiflat}.

 \begin{lemma} \label{lemma:qaf-af}
  Suppose for the function $m \colon [0,\infty) \to \mathbb{R}_0^+$ exists $\alpha \in [0,1)$ such that
 \begin{align}\label{lim}
  \frac{2m}{r} - \alpha = o(r^{-\frac{1}{2}}), \quad \text{as } r \to\infty,
 \end{align}
 and $\rho$ satisfies
 \begin{align}\label{limu}
  \frac{\rho'(r)}{\rho(r)} = o(r^{-\frac{1}{2}}), \quad \text{as } r \to\infty.  
 \end{align}
  Then the solution of the Einstein--Euler system~\eqref{staticEE} is quasi-asymptotically flat with deficit angle $(1-\alpha)\pi$.
 \end{lemma}

 \begin{proof}
 We verify the conditions~\eqref{ss1} of Corollary~\ref{cor:decay}.
 Since $e^{\nu(r)} = \left( \frac{\rho_0}{\rho(r)} \right)^{\frac{K}{1+K}}$, the assumption~\eqref{limu} implies
\[
 \nu'(r) = - \frac{K}{1+K} \frac{\rho'(r)}{\rho(r)} = o(r^{-\frac{1}{2}}).
\]

 The decay assumption \eqref{lim} implies that $e^{2\Lambda} := \lim_{r\to\infty} e^{2\lambda(r)} = \frac{1}{1-\alpha} \geq 1$ exists. Moreover, by the precise definition of little-$o$, for all $C>0$ there exists $r_0>0$ such that $| \frac{2m}{r} - \alpha| \leq C r^{-\frac{1}{2}}$ for $r \geq r_0$. For $r$ sufficiently large, we also have that $|1-\frac{2m}{r}| \geq 1-\alpha - |\alpha-\frac{2m}{r}| \geq \frac{1-\alpha}{2}$.
 Thus, eventually for all $r$ sufficiently large,
 \[
  \left| e^{2\lambda(r)-2\Lambda} - 1 \right| = \left| \frac{\frac{2m}{r}-\alpha}{1-\frac{2m}{r}}\right| \leq \frac{C r^{-\frac{1}{2}}}{\frac{1-\alpha}{2}} = \frac{2C}{1-\alpha} r^{-\frac{1}{2}},
 \]
 which means that $e^{2\lambda(r)-2\Lambda} - 1 = o(r^{-\frac{1}{2}})$. Since $\alpha = 1 - e^{-2\Lambda}$ the claim about the deficit angle follows immediately.
 \end{proof}

 We are now in a position to prove the first main Theorem~\ref{mainthm1} for the linear equation of state.

\begin{proof}[Proof of Theorem~\ref{mainthm1}]
 By Corollary~\ref{cor:asymp}, for every $\varepsilon >0$,
\begin{align*}
 m(r) &= \frac{2K}{(1+K)^2+4K} r + O \left( r^{\frac{1-K}{2(1+K)} + \varepsilon} \right).
\end{align*}
 Therefore, since $-\frac{1+3K}{2(1+K)} + \varepsilon < - \frac{1}{2}$ for $\varepsilon >0$ sufficiently small, as $r \to \infty$,
\begin{align*}
 \frac{2m}{r} - \alpha = O \left( r^{-\frac{1+3K}{2(1+K)} + \varepsilon} \right) = o (r^{-\frac{1}{2}}).
\end{align*}
 Similarly, by Corollary~\ref{cor:asymp} and \eqref{EE2:linear},
\begin{align*}
 \frac{\rho'(r)}{\rho(r)} &= - \frac{1+K}{r-2m} \left( 4\pi r^2 \rho + \frac{m}{rK} \right) \\
&= O(r^{-\frac{1+3K}{2(1+K)} + \varepsilon}) = o(r^{-\frac{1}{2}}).
\end{align*}
 Lemma~\ref{lemma:qaf-af} thus implies quasi-asymptotic flatness, and in particular convergence to $g^\alpha$ as $r\to\infty$.
\end{proof}

\begin{remark}[ADM$\alpha$ mass for the Einstein--Euler system]
In terms of the mass function of the Einstein--Euler system~\eqref{staticEE} the mass function $M$ in Corollary~\ref{cor:MADM} is derived by a transformation of the expression
\begin{align*}
 \left( 1- \frac{2M}{\xi} \right)^{-1} &= e^{2\lambda(r) - 2\Lambda} - (\nu'(r)\tau e^{-\Lambda})^2 \\
 &= (1-\alpha) \left( 1 - \frac{2m}{r} \right)^{-1} - (\nu'(r)\tau e^{-\Lambda})^2 \\
 &\leq (1-\alpha) \left( 1 - \frac{2m}{r} \right)^{-1}
\end{align*}
Thus by Corollary~\ref{cor:MADM}, provided we could prove sufficient decay rates %~\eqref{ss2}
to verify coordinate invariance, we would have that
\begin{align}
m_{\mathrm{ADM}\alpha}(g|_{\Sigma_\tau}) = \lim_{\xi \to\infty} M(\xi)
   &\leq \lim_{\xi \to\infty} \frac{\xi}{2} \left( 1 - \frac{1}{1-\alpha} \left( 1 - \frac{2m(r)}{r} \right) \right) \nonumber \\
   &= \frac{1}{1-\alpha} \lim_{\xi\to\infty} \left( m \frac{\xi}{r} - \frac{\alpha \xi}{2} \right). \label{Mest}
\end{align}
\end{remark}

\begin{remark}[ADM$\alpha$ mass for a linear equation of state]\label{rem:les}
 Since by Corollary~\ref{cor:asymp}
\[
 m(r) - \frac{\alpha}{2} r = O(r^{\frac{1-K}{2(1+K)}+\varepsilon}),
\]
 the estimate \eqref{Mest} would suggest that the ADM$\alpha$ mass of regular perfect fluid solutions to \eqref{EE:linear} is in fact infinite, just like the ADM mass. A more detailed calculation, taking into account the negative term with $\nu'$ and that $\xi = e^{\Lambda} r$, indeed reveals that
\begin{align*}
  M(\xi) &= \frac{\xi}{2} \left[ 1 - \left( 1 - \frac{2m(e^{-\Lambda}\xi)}{e^{-\Lambda}\xi} \right) \left( 1 - \alpha - (\nu'(e^{-\Lambda}\xi)\tau e^{-\Lambda})^2  \left( 1 - \frac{2m(e^{-\Lambda}\xi)}{e^{-\Lambda}\xi} \right) \right)^{-1} \right] \\
          &= \frac{\xi}{2} \left[ 1 - \left( 1 - \alpha + o(\xi^{-\frac{1}{2}}) \right) \left( 1 - \alpha - (1-\alpha) \tau^2 o(\xi^{-1}) \left( 1 - \alpha + (\xi^{-\frac{1}{2}}) \right) \right)^{-1} \right] \\
          &= \frac{\xi}{2} \left[ 1 - \left( 1 - \alpha + o(\xi^{-\frac{1}{2}}) \right) \left( 1 - \alpha + o(\xi^{-1}) \right)^{-1} \right] \\
          &= \frac{\xi}{2} \left[ 1 - 1 + o(\xi^{-\frac{1}{2}}) \right] = o(\xi^{\frac{1}{2}}),
\end{align*}
 and therefore that
\begin{align*}
 m_{\mathrm{ADM}\alpha}(g|_{\Sigma_\tau}) = \lim_{\xi \to\infty} M(\xi) = + \infty \text{ or } - \infty. %\label{Mest}
\end{align*}
 In Section~\ref{ss:singsol}, we have already observed that $m(r) < m_\infty(r)$ for all regular solutions. Therefore, also $M(r) < M_\infty(r) = 0$ for all $r>0$, and the ADM$\alpha$ mass is therefore not only negative but likely also unbounded below. This suggests a negative answer for the question raised in \cite{NS:quasiasymp} whether the ADM$\alpha$ mass of quasi-asymptotically flat metrics is always bounded from below, at least for quasi-asymptotically flat metrics in the sense of Definition~\ref{quasiflat}. However, since we cannot prove that sufficient decay estimates %\eqref{ss2}
 hold, we do in fact not know if we are still in the meaningful situation of Remark~\ref{rem:ADMalpha}, where it is guaranteed that the ADM$\alpha$ mass is a geometric invariant.
\end{remark}

\begin{remark}[Scaled quasi-asymptotic flatness and a ADM$\alpha\beta$ mass]\label{rem:ADMablinear}
 In the optimal quasi-asymptotically flat setting we would set $\beta=0$ even independent of $K$. However, it may be useful to interpret the solutions also in the context of scaled quasi-asymptotic flatness. In the proof of Lemma~\ref{lemma:qaf-af} and the proof of Theorem~\ref{mainthm1} we have observed that, in fact,
  \[
   \nu'(r) = O(r^{-\frac{1+3K}{2(1+K)}+\varepsilon}), \qquad \text{and} \qquad e^{2\lambda(r)-2\Lambda} -1 = O(r^{-\frac{1+3K}{2(1+K)}+\varepsilon}),
  \]
  as $r \to \infty$. Therefore, if $\beta > \frac{2(1+K)}{1+3K}-1 = \frac{1-K}{1+3K}$, then spherically symmetric and static perfect fluid solutions with linear equation of state could also considered to be scaled quasi-asymptotically flat as described in Remark~\ref{higher}. The ADM$\alpha\beta$ satisfies
  \begin{align*}
   m_{\mathrm{ADM}\alpha\beta}(g|_{\Sigma_\tau}) &= \frac{1}{2} \lim_{\xi\to\infty} (1+\beta) (1+O(\xi^{-1})+\tau^2 O(\xi^{-1})-1)\xi \\
                                                 &= O(1) +\tau^2 O(1) < \infty
  \end{align*}
 according to \eqref{mab} for an optimal $\beta$ (which we indeed expect to be $\frac{1-K}{1+3K}$). Here, we also have to change $\alpha$, which is then given by
 \begin{align*} 1-\alpha = (1+\beta)^{-2} e^{-2\Lambda} < \frac{(1+3K)^2}{4 (1+K)^2} \frac{(1+K)^2}{(1+K)^2+4K} = \frac{(1+3K)^2}{4((1+K)^2+4K)}. \end{align*}
 Since $K \in (0,1)$ implies $5K^2 < 6K$, the denominator is larger than the numerator, and hence also the rescaled $\alpha = 1-(1+\beta)^{-2} e^{-2\Lambda} > \frac{3+6K-5K^2}{4((1+K)^2+4K)}$ is contained in the interval $(0,1)$.
\end{remark}

%%%%%%%%%%%%%%%%%%%%%%%%%%%%%%%%%%%%%%%%%%%%%%%%%%%%%%%%%%%%%%%%%%%%%%%%%%%%%%%%%%%%%%%%
%%%%%% SECTION 4 %%%%%%%%%%%%%%%%%%%%%%%%%%%%%%%%%%%%%%%%%%%%%%%%%%%%%%%%%%%%%%%%%%%%%%%
%%%%%%%%%%%%%%%%%%%%%%%%%%%%%%%%%%%%%%%%%%%%%%%%%%%%%%%%%%%%%%%%%%%%%%%%%%%%%%%%%%%%%%%%

\section{Perfect fluids with polytropic equation of state}
\label{sec4}

 In their dynamical systems approach to the spherically symmetric static Einstein--Euler system~\eqref{staticEE}, Heinzle, R\"ohr and Uggla~\cite{HRU} made use of the quantities
 \[
 \Gamma_N (p) := \frac{\rho}{p} \frac{dp}{d\rho}, \qquad \sigma(p) := \frac{p}{\rho},
\]
 defined in terms of the equation of state $p=p(\rho)$. In \cite{HRU}*{Theorem 5.1} it was shown that all regular solutions of \eqref{staticEE} have infinite masses and infinite radii if $\Gamma_N \leq \frac{6}{5}$ and $\sigma \leq 1$. The assumption on $\sigma$ is related to the dominant energy condition.
 
 For polytropic equations of state \eqref{polytropicEOS}, i.e., $p = K \rho^\frac{n+1}{n}$,
 we obtain that 
 \[ \Gamma^\mathrm{poly}_N(p) = \frac{n+1}{n}, \qquad \sigma^\mathrm{poly}(p) = K \rho^\frac{1}{n} = K^{\frac{n}{n+1}} p^{\frac{1}{n+1}}.\]
 Clearly, for $n\geq 5$, the condition $\Gamma^\mathrm{poly}_N \leq \frac{6}{5}$
 is satisfied, while the second condition $\sigma \leq 1$ is not satisfied in the high-pressure regime. Heinzle, R\"ohr and Uggla~\cite{HRU} therefore considered equations of state that are linear for high pressures in \cite{HRU}, which leads to the analysis of so-called barotropic equations of state.
 
 In a more detailed analysis, Nilsson and Uggla~\cite{NU:polytropic}*{Section 2} explain that the spherically symmetric Einstein--Euler system with power-law polytropic equations of state $p = K \rho^{\frac{n+1}{n}}$, i.e., the system
\begin{subequations} \label{EE:poly}
 \begin{align}
  m_r &= 4 \pi r^2 \rho, \label{EE1:poly} \\
  \rho_r &= - \frac{n}{n+1} \frac{(1+K \rho^{\frac{1}{n}}) \rho^{\frac{n-1}{n}}}{r-2m} \left( 4 \pi r^2 \rho^{\frac{n+1}{n}} + \frac{m}{rK}\right), \label{EE2:poly}
 \end{align}
\end{subequations}
 yields finite radius solutions if $0 < n < 5$ and if the central density $\rho_0$ is small. See also \cite{RS:static}*{Theorem 4} where a result also for generalized power-law polytropic equations of state was obtained for $1<n<5$. In the case of $0<n\leq 3$, there exists a global sink $P_2$ where all orbits end (see also \cite{Mak:static}*{Theorem 1}). If $3<n<5$, then the majority of orbits still end at $P_2$, but orbits ending at $P_1$ (which have finite masses but infinite radii) and $P_4$ (which have infinite masses and radii) also occur. When $n \gtrsim 3.339$ it was shown numerically that there is at least one solution ending at $P_1$. At $n\approx 3.357$ and $n \approx 4.414$ Nilsson and Uggla obtained solutions with infinite masses and radii corresponding to $P_4$. The dynamical behavior turns out to be quite complicated and is not yet fully understood from an analytical point of view. For more details, see \cite{NU:polytropic}*{Sections 2.5--2.7}. For $n\geq 5$ all relativistic regular models have infinite radii and masses, and spiral around the Tolman orbit $P_4$, which is associated with a special nonregular Newtonian solution that is not known in exact form.

\subsection{The initial value problem}\label{sec:ivp-polytropic}

 In what follows, we mainly restrict our attention to the power-law polytropic equation of state~\eqref{polytropicEOS} with polytropic index $n>5$. Fix $K \in (0,1]$ and a central density $\rho_0>0$. By \cite{RS:static}*{Theorem 2} exists a unique, smooth and positive global solution $(m,\rho)$ of \eqref{EE:poly} such that
\[
 \lim_{r\to 0} m(r) = 0, \qquad \lim_{r\to0} \rho(r) = \rho_0.
\]
The solution has infinite extent due to our discussion in Section~\ref{ssec:extent}, with $\rho(r) \to 0$ as $r \to \infty$. The asymptotic behavior as $r \to 0$ is given by the Taylor series expansion which reads
\begin{align*}
  m(r) &= \frac{4\pi}{3} \rho_0 r^3  + O(r^ 5),\\
 \rho(r) &= \rho_0 - 2\pi \frac{n}{n+1} \frac{(1+K\rho_0^{\frac{1}{n}})(1+3K\rho_0^{\frac{1}{n}})}{K\rho_0^{\frac{1}{n}}} \rho_0^2 r^2 + O(r^4) ,
\end{align*}
 due to l'H\^{o}pital's rule.

\subsection{The asymptotic behavior of solutions}

 Due to the dynamical systems analysis \cite{NU:polytropic} we know that, asymptotically for $r \to \infty$, some solutions of \eqref{EE:poly} with $3<n<5$ and all solutions with $n\geq 5$ converge to a fixed point $P_4$ (which corresponds to $B_4$ in \cite{HRU}). In terms of $m$ and $\rho$, this provides us with some control on the asymptotic behavior of (some) solutions. We are primarily interested in polytropes with index $n>5$, but some results also hold for the (unstable) infinite solutions with index $3 < n< 5$.

\subsubsection{The dynamical system and its fixed points}\label{ss:polysys}
 The formulation of Nilsson and Uggla~\cite{NU:polytropic} for the Einstein--Euler system  in the case of the power-law polytropic equation of state makes use of the transformation
\begin{align}\label{UV:m}
 U = \frac{4 \pi r^2 \rho}{4 \pi r^2 \rho + \frac{m}{r}}, \quad V = \frac{\frac{m}{r}}{K\rho^{\frac{1}{n}}+\frac{m}{r}}, \quad y = \frac{K \rho^{\frac{1}{n}}}{K \rho^{\frac{1}{n}} + 1}.%\frac{m/r}{m/r + p/\rho},
\end{align}
which leads to the dynamical system
\begin{subequations} \label{UV}
\begin{align}
 \dot U &= U (1-U) [(1-y)(3-4U)F - \frac{n}{n+1} G], \label{UV1} \\
 \dot V &= V (1-V) [(1-y)(2U-1)F + \frac{1}{n+1} G], \label{UV2} \\
 \dot y &= - \frac{1}{n+1} y (1-y) G.
\end{align}
\end{subequations}
 with
\[
 F = (1-V)(1-y)-2yV, \qquad G = V[(1-U)(1-y)+yU]. %, \qquad \Gamma = 1+ \frac{1}{n}. 
\]
 The differentiation in \eqref{UV} is with respect to a new independent variable (indirectly related to the geometry) and $r^2$ is given by
 \[
 r^2 = \frac{K^n}{4\pi} \frac{U}{1-U}\frac{V}{1-V} \left( \frac{1-y}{y} \right)^{n-1}.  
 \]
 According to the numerical analysis in \cite{NU:polytropic}*{Section 2}, if $3<n\lesssim 3.339$, all regular orbits end at the fixed point $P_2$ (which is the only hyperbolic sink in this case and leads to solutions with finite masses and radii). For $n \gtrsim 3.339$ isolated orbits also end at the equilibrium points $P_1$ (solutions with finite masses but infinite radii) and $P_4$ (solutions with infinite masses and infinite radii). The latter fixed point $P_4$ is given by
\begin{align}\label{UV:n}
 U_4 = \frac{n-3}{2(n-2)}, \quad V_4 = \frac{2(n+1)}{3n+1}, \quad y_4 = 0, \qquad \text{for }n>3,
\end{align}
 with eigenvalues\footnote{Note that the term $\frac{7n^2}{2}-11n-\frac{1}{2}$ under $\sqrt{}$ is nonnegative only for $n \geq 3.18767$. Otherwise the eigenvalues are, in fact, real.}
\begin{align}
 \lambda_\pm &= - \frac{(n-1)(n-5)}{4(n-2)(1+3n)} \pm i \frac{n-1}{4(n-2)(1+3n)} \sqrt{\frac{7n^2}{2}-11n-\frac{1}{2}}, \label{lambda2} \\
 \lambda_0 &= - \frac{n-1}{(n-2)(1+3n)}. \nonumber
\end{align}
 We see that $- \frac{1}{12} \leq \Re \lambda_\pm = -\frac{(n-1)(n-5)}{4(n-2)(1+3n)} < 0$ if and only if $n>5$, in which case $P_4$ becomes a hyperbolic sink and all orbits end at $P_4$, leading to solutions with infinite radii and masses \cite{NU:polytropic}.

\subsubsection{Asymptotic stability for solutions of \eqref{EE:poly} with polytropic index $n>5$}

 The above properties and the Hartman--Grobman theorem (see, for example, \cite{R}*{Theorem 5.3}) imply that the behavior of the dynamical system~\eqref{UV:n} near the fixed point $P_4$ is qualitatively given by its linearization. The flow of \eqref{UV:n} is $C^1$-conjugate to the affine flow $s \mapsto P_4 + e^{A_4 s}(x-P_4)$, where the linearization around $P_4$ is given by
 \[
  A_4 %= %DF|_{P_4}
  = \begin{pmatrix}
-\frac{(n-3) (n-1)}{
   2 (n-2) (1 + 3 n)} & -\frac{(n-3) (n-1) n (1 + 3 n)}{
   8 (n-2)^3 (1 + n)} & - \frac{3(n-3)(n-1)n^2}{2(n-2)^3(1+3n)} \\ 
   \frac{4 (n-2) (n-1) (n+1)}{(1 + 3 n)^3} & \frac{n-1}{(n-2) (1 + 3 n)} & \frac{12(n-1)(n+1)}{(n-2)(1+3n)^3} \\
 0 & 0 & -\frac{n-1}{(n-2) (1 + 3 n)}
                    \end{pmatrix}.
 \]
 %Unlike in the Newtonian limit, the system does not decouple. However,
 As observed in \cite{NU:polytropic}, on the subset $\{y=0\}$ the relativistic equations~\eqref{UV1}--\eqref{UV2} and the corresponding two Newtonian equations and coincide.  Since, however, we cannot directly relate the new independent variable to $r$, we cannot compute a convergence rate of $m$ and $\rho$ as $r \to \infty$. We merely compute the leading order term based on the results of Nilsson and Uggla~\cite{NU:polytropic}.
 
 \begin{proposition}[Asymptotic stability in terms of $(m,\rho)$]\label{cor:asymp-poly}
 Fix $K \in (0,1]$ and $n > 5$. The asymptotic behavior of solutions $(m,\rho)$ to the system~\eqref{EE:poly} with initial data $\rho_0>0$ is
\begin{align*}
 m(r) &= 
\sqrt[n-1]{\frac{2^{n-2} K^{n}}{\pi} \frac{(n-3)(n+1)^n}{(n-1)^{n+1}}} r^{\frac{n-3}{n-1}} + o\big(r^{\frac{n-3}{n-1}}\big), \\
 \rho(r) &= \sqrt[n-1]{\frac{K^n}{2^n\pi^n} \frac{(n+1)^n(n-3)^n}{(n-1)^{2n}}} r^{-\frac{2n}{n-1}} + o\big(r^{-\frac{2n}{n-1}}\big),
\end{align*}
as $r\to\infty$.
 \end{proposition}

 \begin{proof}
  This follows directly from the analysis of the dynamical system \eqref{UV} in \cite{NU:polytropic}, since in the case of $n>5$ the variables $(U,V,y)$ converge to $(U_4,V_4,y_4)$ given in \eqref{UV:n}.
  By \eqref{UV:m} we deduce that
\begin{align*}
 \rho^{\frac{n-1}{n}} &= \frac{K}{4\pi r^2} \frac{U}{1-U} \frac{V}{1-V},
\end{align*}
 hence, the leading order term of $\rho$ is
 \begin{align*}
 \rho_4 = \left( \frac{K}{2\pi} \frac{(n+1)(n-3)}{(n-1)^2} \right)^{\frac{n}{n-1}} r^{-\frac{2n}{n-1}}. 
\end{align*}
 Similarly, \eqref{UV:m} implies that
\begin{align*}
 m = 4\pi r^3 \frac{U-1}{U}
\end{align*}
 with leading order term
\begin{align*}
 m_4 &= \frac{2(n+1)K}{n-1}  \left( \frac{K}{2\pi} \frac{(n+1)(n-3)}{(n-1)^2} \right)^{\frac{1}{n-1}} r^{-\frac{2}{n-1}+1} 
       = \left( \frac{2^{n-2} K^n}{\pi} \frac{(n+1)^n(n-3)}{(n-1)^{n+1}} \right)^{\frac{1}{n-1}} r^{\frac{n-3}{n-1}}. 
       \qedhere
\end{align*}
\end{proof}

\begin{remark}
 Be aware that, unlike in the case of the linear equation of state (cf.\ Section~\ref{ss:singsol}), the leading order terms of the variables $(m,\rho)$, i.e.,
\begin{align*}
 `` m_4(r)\textquotedblright &=  \sqrt[n-1]{\frac{2^{n-2} K^{n}}{\pi} \frac{(n-3)(n+1)^n}{(n-1)^{n+1}}} r^{\frac{n-3}{n-1}},\\
 ``\rho_4 (r)\textquotedblright &= \sqrt[n-1]{\frac{K^n}{2^n\pi^n} \frac{(n+1)^n(n-3)^n}{(n-1)^{2n}}} r^{-\frac{2n}{n-1}},
\end{align*}
 do \emph{not} yield a (singular) solution to the original system~\eqref{EE:poly} itself, but merely represents the asymptotic behavior of regular solutions as $r \to \infty$. Due to the third component of $P_4$ only the behavior near $y_4=0$, i.e., at $r = \infty$, is described. One cannot apply Proposition~\ref{cor:asymp-poly} to the limit $n \to \infty$ (corresponding to the linear equation of state) since then $y = \frac{p}{p+\rho} = \frac{K}{1+K} = \textrm{const.} > 0$.
\end{remark}

\subsection{Scaled quasi-asymptotic flatness and ADM$\alpha\beta$ mass}

 In Proposition~\ref{cor:asymp-poly} we have observed that, for static fluids with polytropic index $n>5$, the mass function behaves like $m(r) \sim C r^{\frac{n-3}{n-1}}$ for some constant $C = C(n,K)$ as $r \to \infty$. The expression \eqref{ADMss} would therefore yield an ADM mass
\[
 m_\textrm{ADM}(g) =  \lim_{r\to\infty} m(r) = \infty.
\]
 However, although
\[
 \lim_{r\to\infty} e^{2\lambda(r)} = \lim_{r\to\infty} \left( 1 - \frac{2m(r)}{r} \right)^{-1} = \lim_{r\to\infty} \frac{1}{1-Cr^{-\frac{2}{n-1}}} = 1,
\]
 as in the asymptotically flat situation, we are in a situation where the ADM mass is not coordinate invariant because $a(r)-1 = \left(1 - \frac{2m}{r} \right)^{-1}-1 = O(r^{-\frac{2}{n-1}})$ rather than $o(r^{-\frac{1}{2}})$.

 Is the solution quasi-asymptotically flat in the sense of Definition~\ref{quasiflat}? Because of the above, we would have to set $\alpha=0$, but then again violate condition~\eqref{lim} in Lemma~\ref{lemma:qaf-af} which also requires that $\frac{2m}{r}-\alpha = \frac{2m}{r} = o(r^{-\frac{1}{2}})$. 
 On the other hand, we immediately see from Corollary~\ref{cor:asymp-poly} that
\[
 \frac{2m}{r^{\frac{n-3}{n-1}}} \sim \sqrt[n-1]{\frac{2^{n-2} K^{n}}{\pi} \frac{(n-3)(n+1)^n}{(n-1)^{n+1}}} \qquad \text{as}~ r \to\infty. %= o(r^{-\frac{1}{2}}).
\]
 This nonlinear scaling in the radial direction is different from the quasi-asymptotically flat case for fluids with linear equation of state (although, as $n \to \infty$, $\frac{n-3}{n-1}$ approaches $1$). To accommodate this behavior we introduced the notion of scaled quasi-asymptotic flatness in Section~\ref{sec2}.
 It remains to verify that perfect fluids with polytropic equations of state with index $n>5$ are indeed scaled quasi-asymptotically flat in the sense of Definition~\ref{def:sqaf}. We prove our second main Theorem~\ref{mainthm2} for the polytropic equation of state.

\begin{proof}[Proof of Theorem~\ref{mainthm2}]
 We verify the conditions~\eqref{ss3} of Proposition~\ref{prop:sqaf}. By Proposition~\ref{cor:asymp-poly} the asymptotic behavior of $m$ and $\rho$ is
\begin{align*}
 m(r) &=
\sqrt[n-1]{\frac{2^{n-2} K^{n}}{\pi} \frac{(n-3)(n+1)^n}{(n-1)^{n+1}}} r^{\frac{n-3}{n-1}} + o \big( r^{\frac{n-3}{n-1}} \big), \\ 
 \rho(r) &= \sqrt[n-1]{\frac{K^n}{2^n\pi^n} \frac{(n+1)^n(n-3)^n}{(n-1)^{2n}}} r^{-\frac{2n}{n-1}} + o \big( r^{-\frac{2n}{n-1}} \big). 
\end{align*}
 Hence, the differential equations \eqref{EE2:poly} and \eqref{inteq1} imply that
\begin{align*}
 \nu'(r) &= - \frac{p'(r)}{p(r)+\rho(r)} = - \frac{K \frac{n+1}{n} \rho^{\frac{1}{n}} \rho_r}{\rho (K \rho^{\frac{1}{n}} + 1)} \\
         &= \frac{K}{r (1-\frac{2m}{r})} \left( 4 \pi r^2 \rho^{\frac{n+1}{n}} + \frac{m}{rK} \right) \\
         &= \frac{K}{r(1 - O(r^{-\frac{2}{n-1}}))} \left( O(r^{-\frac{4n}{n-1}}) + O(r^{-\frac{2}{n-1}}) \right) = O(r^{-\frac{n+1}{n-1}})
\end{align*}
 and
\begin{equation*} 
 e^{2\lambda(r)-2\Lambda} - 1 = \left( 1 - \frac{2m}{r} \right)^{-1} - 1 = \frac{1}{1-O(r^{-\frac{2}{n-1}})} - 1 =  O(r^{-\frac{2}{n-1}}).
\end{equation*}
Since $- \frac{n+1}{n-1} \leq - \frac{2}{n-1}$, the solution is AF$\alpha\beta$ with $\beta > \frac{n-5}{4}$ and $\alpha = 1-\frac{1}{(1+\beta)^2} > 1 - \left(\frac{4}{n-1}\right)^2>0$ according to Proposition~\ref{prop:sqaf}.
\end{proof}

\begin{remark}[ADM$\alpha\beta$ mass]\label{rem:ADMabpoly}
 We believe that it is possible to set $\beta = \frac{n-5}{4}$ and $1-\alpha = \left( \frac{4}{n-1} \right)^2$ in Theorem~\ref{mainthm2}. However, in view of Remark~\ref{higher}, we consider $\beta = \frac{n-3}{2}$, which is negative for all values of $n$ that definitely lead to solutions with finite extent and is greater than $1$ for solutions with definitely infinite extent. The corresponding ADM$\alpha\beta$ mass \eqref{mab} from Remark~\ref{rem:ADMalphabetass} is then
  \begin{align*}
   m_{\text{ADM}\alpha\beta} (g|_{\Sigma_\tau}) &= \frac{1}{2} \lim_{\xi \to\infty} (1+\beta) \left( (1+a_{\xi\xi}(\tau,\xi)) - 1 \right) \xi, 
  \end{align*}
 where
 \begin{align*}
  a_{\xi\xi} &= - 1 + e^{2\lambda(r)-2\Lambda} - e^{-2\Lambda} \nu'(r)^2 \tau^2 = \left( 1 - \frac{2m}{r} \right)^{-1} - 1 + O(r^{-\frac{2(n+1)}{n-1}}) \\ &= \frac{2C(n,K)}{\xi - 2C(n,K) + o(\xi^{-\frac{1}{1+\beta}})} + O(\xi^{-2})
 \end{align*}
 with leading order coefficient
 $
  C(n,K) %= \sqrt[n-1]{\frac{2^{n-2} K^{n}}{\pi} \frac{(n-3)(n+1)^n}{(n-1)^{n+1}}}
 $
 of $m$. Thus
 \begin{align*}
   m_{\text{ADM}\alpha\beta} (g|_{\Sigma_\tau}) &= \frac{1}{2} \lim_{\xi \to\infty} (1+\beta) \left( \frac{2C(n,K)}{\xi - 2C(n,K) + o(\xi^{-\frac{1}{1+\beta}})} +  O(\xi^{-2}) \right) \xi \\
   &= (1+\beta)C(n,K) \\
   &= (1+\beta)\sqrt[n-1]{\frac{2^{n-2} K^{n}}{\pi} \frac{(n-3)(n+1)^n}{(n-1)^{n+1}}} \\
   &= \sqrt[n-1]{\frac{K^{n}}{2\pi} \frac{(n-3)(n+1)^n}{(n-1)^{2}}}.
  \end{align*}
\end{remark}

%%%%%% ACKNOWLEDGMENTS %%%%%%

\section*{Acknowledgments}

 The authors would like to thank Patryk Mach for discussions on the Lane--Emden equation and Naresh Dadhich for pointing out the conformal asymptotically flat geometry. We would also like to thank Claes Uggla for correspondence regarding their dynamical systems approach, Piotr Chru{\'s}ciel for information on the asymptotically flat situation, Anna Sakovich for pointing out references related to notions of mass and Shadi Tahvildar-Zadeh for general discussions on the equations of state.

 A.B.\ acknowledges financial support during an extended research stay at the Riemann Center for Geometry and Physics at the Leibniz University Hannover, during which this project was initiated. The major part of the project was carried out at the University of Bonn and supported by the Sofja Kovalevskaja award of the Humboldt Foundation endowed by the German Federal Ministry of Education and Research (held by Roland Donninger). The authors also gratefully acknowledge financial support in connection with a research stay at the Erwin Schr\"odinger Institute in Vienna during the ''Geometry and Relativity`` program in 2017, where the project was finalized.

%%%%%% BIBLIOGRAPHY %%%%%%

\end{document}